\documentclass[11pt,letterpaper]{article}

\usepackage{amssymb,amsfonts,amsmath,amsthm}
\usepackage{graphicx,float,psfrag,epsfig}
\usepackage{wrapfig}

\newtheorem{propo}{Proposition}[section]
\newtheorem{lemma}[propo]{Lemma}
\newtheorem{definition}[propo]{Definition}
\newtheorem{coro}[propo]{Corollary}
\newtheorem{thm}[propo]{Theorem}

\def\hG{\widehat{G}}
\def\hY{\widehat{Y}}
\def\hZ{\widehat{Z}}
\def\hprob{\widehat{\mathbb P}}
\def\hE{\widehat{\mathbb E}}

\def\graph{{\cal G}}
\def\hZ{\widehat{Z}}

\def\ve{\varepsilon}

\def\<{\langle}
\def\>{\rangle}
\def\ed{\stackrel{{\rm d}}{=}}
\def\da{{\partial a}}

\def\di{{\partial i}}
\def\dj{{\partial j}}
\def\dn{{\partial n}}
\def\d2n{{\partial^2 n}}

\def\naturals{{\mathbb N}}

\def\reals{{\mathbb R}}

\def\de{{\rm d}}

\def\tE{\widetilde{\mathbb E}}
\def\E{{\mathbb E}}

\def\prob{{\mathbb P}}
\def\tprob{\widetilde{\mathbb P}}
\def\hprob{\widehat{\mathbb P}}
\def\qprob{{\mathbb Q}}
\def\Var{{\rm Var}}

\def\cX{{\cal X}}
\def\cS{{\cal S}}
\def\hcS{\widehat{\cal S}}
\def\hp{\widehat{p}}
\def\meas{{\sf M}}

\def\sE{{\sf E}}
\def\sP{{\sf P}}

\def\me{\nu}
\def\mh{\widehat{\nu}}

\def\ind{{\mathbb I}}
\def\eps{\epsilon}
\def\Sp{{\sf S}}
\def\Qp{{\sf Q}}

\def\Ball{{\sf B}}
\def\cBall{\overline{\sf B}}
\def\dBall{{\sf D}}
\def\Tree{{\sf T}}
\def\Comp{{\mathfrak C}}

\def\H{{\mathbb H}}
\def\oX{X'}
\def\ox{x'}
\def\uth{\underline{\theta}}

\def\Fbp{{\sf F}^n}
\def\Fbps{{\sf F}^{n_*}}
\def\Fde{{\sf F}^{\infty}}

\def\sTV{\mbox{\tiny\rm TV}}
\def\1t{{\tt 1}}
\def\0t{{\tt 0}}
\def\ogamma{\overline{\gamma}}

\newcommand{\eqnsection}{\renewcommand{\theequation}{\thesection.\arabic{equation}}
     \makeatletter \csname @addtoreset\endcsname{equation}{section}\makeatother}

\begin{document}
\eqnsection

\title{Estimating Random Variables from Random Sparse Observations}

\author{Andrea Montanari\thanks{Departments of Electrical Engineering 
and Statistics,
Stanford University,
{\tt montanari@stanford.edu}}}

\date{\today}
\maketitle

\abstract{Let  $X_1,\dots, X_n$ be a collection of iid discrete 
random variables, and $Y_1,\dots, Y_m$ a set of noisy observations
of such variables. Assume each observation $Y_a$ to be a random function of
some a random subset of the $X_i$'s, and consider the conditional
distribution of $X_i$ given the observations, namely 
$\mu_i(x_i)\equiv\prob\{X_i=x_i|Y\}$
(\emph{a posteriori probability}). 

We establish a general decoupling principle among the $X_i$'s, as well 
as a relation between the distribution of $\mu_i$,
and the fixed points of the associated density evolution operator.
These results hold asymptotically in the large system limit,
provided the average number of variables an observation depends on
is bounded.
We discuss the relevance of our result to a number of applications, 
ranging from
sparse graph codes, to multi-user detection, to group testing.}
%
%
\section{Introduction}

Sparse graph structures have proved useful in a number of information 
processing tasks, from channel coding \cite{RiU05}, 
to source coding \cite{CaireEtAl}, to sensing
and signal processing \cite{Donoho,CandesRombergTao}. Recently
similar design ideas have been proposed for code division multiple access
(CDMA) communications \cite{MonTse,TanakaSparse06,SaadSparse}, and group testing (a classical 
technique in statistics) \cite{MezardGroup}. 

The computational problem underlying many of these developments
can be described as follows: \emph{infer the values of a large collection 
of random variables,
given a set of constraints, or observations, that induce relations among them.}
While such a task is generally computationally hard 
\cite{Berlekamp,NP_CDMA}), sparse graphical structures allow for 
low-complexity algorithms (for instance iterative message passing
algorithms as belief propagation) that were revealed to be
very effective in practice. A precise analysis of these 
algorithms and of their gap to optimal (computationally intractable)
inference is however a largely open problem.

In this paper we consider an idealized setting in which 
we aim at estimating $n$ iid discrete random variables $X=(X_1,\dots,X_n)$ 
based on noisy observations. We will focus on the large system
limit $n\to\infty$, with the number of observations scaling like $n$.
We further restrict our system to be \emph{sparse} in the sense that 
each observation depends on a bounded (on average) number of variables.
A schematic representation is given in Fig.~\ref{fig:Factor}.

If $i\in[n]$, and $Y$ denotes collectively the observations,
a sufficient statistics for estimating $X_i$ is 
\begin{eqnarray}
\mu_i(x_i) = \prob\{X_i=x_i|Y\}\, .
\end{eqnarray}
This paper establishes two main results: an asymptotic decoupling among
the $X_i$'s, and a characterization of the asymptotic distribution
of $\mu_i(\,\cdot\,)$ when $Y$ is drawn according to the source and channel 
model. In the remainder of the introduction we will discuss a few (hopefully)
motivating examples, and we will give an informal summary of our results.
Formal definitions, statements and proofs can be found in Sections
\ref{sec:Def} to \ref{sec:MainProof}.
%
%
\subsection{Motivating examples}
\label{sec:Examples}

In this Section we present a few examples that fit within the 
mathematical framework developed in the present paper.
The main restrictions imposed by this framework are: 
$(i)$ The `hidden variables' $X_i$'s are independent;
$(ii)$ The bipartite graph $G$ connecting hidden variables and observations
lacks  any geometrical structure. 

Our results crucially rely on 
these two features. Some further technical assumptions 
will be made that partially rule out some of these examples
below. However  we expect these assumptions to be removable by 
generalizing the arguments presented in the next sections.

\begin{figure}[t]
\center{\includegraphics[width=9.0cm]{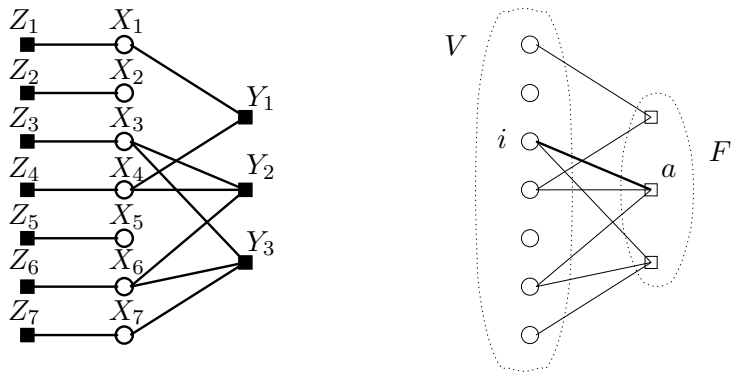}}
\put(-260,130){$Z_1$}
\put(-222,130){$X_1$}
\put(-260,111){$Z_2$}
\put(-222,111){$X_2$}
\put(-260,93){$Z_3$}
\put(-222,93){$X_3$}
\put(-260,74){$Z_4$}
\put(-222,74){$X_4$}
\put(-260,56){$Z_5$}
\put(-222,56){$X_5$}
\put(-260,39){$Z_6$}
\put(-222,39){$X_6$}
\put(-260,20){$Z_7$}
\put(-222,20){$X_7$}
\put(-170,101){$Y_1$}
\put(-170,75){$Y_2$}
\put(-170,46){$Y_3$}
\put(-95,120){$V$}
\put(5,80){$F$}
\put(-75,85){$i$}
\put(-13,74){$a$}
\caption{{\small Factor graph representation of a simple sparse observation 
systems with $n=7$ hidden variables $\{X_1,\dots,X_7\}$ and $m=3$
`multi-variable' observations $\{Y_1,\dots,Y_3\}$. On the right: the 
bipartite graph $G$. Highlighted is the edge $(i,a)$.}}
\label{fig:Factor}
\end{figure}
%
\vspace{0.25cm}

{\bf Source coding through sparse graphs.}
Let $(X_1, \dots, X_n)$ be iid Bernoulli$(p)$. 
Shannon's theorem implies that such a vector can be stored in 
$nR$ bits for any $R>h(p)$ (with $h(p)$ the binary entropy function), 
provided we allow for a vanishingly small failure probability.
The authors of 
Refs.~\cite{Murayama,MurayamaISIT,CaireEtAl} proposed to implement this 
compression through a sparse linear transformation. Given a source 
realization $X=x = (x_1,\dots,x_n)$, the stored vector reads 
\begin{eqnarray*}
y = \H x\;\;\;\;\; \mod 2\, ,
\end{eqnarray*}
with $\H$ a sparse $\{0,1\}$ valued 
random matrix of dimensions $m\times n$, and $m=nR$.
According to our general model, each of the coordinates
of $y$ is a function (mod $2$ sum) of a bounded (on average)
subset of the source buts $(x_1,\dots,x_n)$.

The $i$-th information bit can be reconstructed
from the stored information by computing the conditional distribution 
$\mu_i(x_i)=\prob\{X_i=x_i|Y\}$.
In practice, belief propagation provides a rough estimate of $\mu_i$.
Determining the distribution of $\mu_i$ (which is the main topic of the 
present paper) allows to determine
the optimal performances (in terms of bit error rate) of such a system.

\vspace{0.25cm}

{\bf Low-density generator matrix (LDGM) codes.}
We take $(X_1,\dots, X_n)$ iid Bernoulli$(1/2)$,
encode them in a longer vector $\oX=
(\oX_1,\dots,\oX_m)$ via the mapping $\ox = \H x \mod 2$, 
and transmit the encoded bits through a noisy memoryless channel, 
thus getting output $(Y_1,\dots,Y_m)$
\cite{LubyTransform}. One can for instance think of
a binary symmetric channel BSC$(p)$, 
whereby $Y_a= \oX_a$ with probability $1-p$
and $Y_a=\oX_a\oplus\1t$ with probability $1-p$. 
Again, decoding can be realized through a belief propagation estimate
of the conditional probabilities $\mu_i(x_i)= \prob\{X_i=x_i|Y\}$.

If the matrix $\H$ is random and sparse, this problem fits 
in our framework with the information (uncoded) bits $X_i$'s being hidden
variables, while the $y_a$'s correspond to observations.
\vspace{0.25cm}

{\bf Low-density parity-check (LDPC) codes.}
With LDPC codes, one sends through a noisy channel a codeword
$X=(X_1,\dots,X_n)$  that is a uniformly random
vector in the null space of a random sparse matrix $\H$
\cite{GallagerThesis,RiU05}.
While in general this does not fit our setting, one can construct
an equivalent problem (for analysis purposes) which does,
provided the communication channel is binary memoryless symmetric, 
say BSC$(p)$.

Within the equivalent problem $(X_1,\dots,X_n)$ are iid
Bernoulli$(1/2)$ random bits. Given one realization $X=x$ of these
bits, one computes its syndrome $y=\H x \mod 2$ and transmits it through
a noiseless channel. Further, each of the $X_i$'s is transmitted
through the original noisy channel (in our example BSC$(p)$) yielding
output $Z_i$. If we denote the observations collectively as
$(Y,Z)$, it is not hard to show that the conditional
probability $\mu_i(x_i)=\prob\{X_i=x_i|Y,Z\}$ has in fact the same 
distribution of the a posteriori probabilities in the original LDPC model.

Characterizing this distribution allows to determine the 
information capacity of such coding systems, and their performances under
MAP decoding \cite{Montanari05,KudekarMacris06,KoradaKudekarMacris07}.
\vspace{0.25cm}

{\bf Compressed sensing.}
In compressed sensing the real vector $x=(x_1,\dots,x_n)\in\reals^{n}$ 
is measured through a set of linear projections $y_1=h_1^{T}x$,
$\dots, y_m=h_m^Tx$. In this literature no assumption 
is made on the distribution of $x$, which is only constrained to be
sparse in a properly chosen basis \cite{Donoho,CandesRombergTao}. 
Further, unlike in our setting, the vector components $x_i$ do not 
belong to any finite alphabet. However, some applications justify the 
study of a probabilistic version, whereby the basic variables are 
quantized. An example is provided by the next item.
\vspace{0.25cm}

{\bf Network measurements.} 
The size of flows in the Internet can vary from a few  (as 
in acknowledgment messages) to several million packets (as in content 
downloads). Keeping track of the sizes of flows passing through a router can be
useful for a number of reasons, such as billing, or security, 
or traffic engineering \cite{Varghese03}. 

Flow sizes can be  modeled as iid random integers $X=(X_1,\dots,X_n)$. 
Their common distribution is often assumed to be a heavy-tail one.
As a consequence, the largest flow is typically of size $n^{a}$
for some $a>0$. It is
therefore highly inefficient to keep a separate counter of capacity 
$n^a$ for each flow. It was proposed in \cite{CountAllerton07}
to store instead a shorter vector $Y= \H X$, with $\H$ 
a properly designed sparse random matrix. The problem of 
reconstructing the $X_i$'s is, once more, analogous to the 
above.
\vspace{0.25cm}

{\bf Group testing.} Group testing was proposed during World War II
as a technique for reducing the costs of syphilis tests in the 
army by simultaneously testing \emph{groups} of soldiers
\cite{GTesting}.
The variables $X=(X_1,\dots,X_n)$ represent the individuals status
($1=$ infected, $0=$ healthy) and are modeled as iid Bernoulli$(p)$
(for some small $p$). Test $a\in\{1,\dots, m\}$ involves a subset 
$\da\subseteq [n]$ of the individuals and returns positive value
$Y_a=1$ if $X_i=1$ for some $i\in\da$ and $Y_a=0$ 
otherwise\footnote{The problem can be enriched by allowing for a small 
false negative (or false positive) probability.}.
It is interesting to mention that the problem has a connection with 
random multi-access channels, that was exploited in \cite{MAccess1,MAccess2}.

One is interested in the conditional probability for
the $i$-th individual to be infected given the observations:
$\mu_i(1) = \prob\{X_i=1|Y\}$. Choices of the groups (i.e. of the subsets
$\da$) based on random graph structures where recently studied 
and optimized in \cite{MezardGroup}. 

\vspace{0.25cm}

{\bf Multi-user detection.} In a general vector channel,
one or more users communicate symbols $X=(X_1,\dots,X_n)$
(we assume, for the sake of simplicity, perfect synchronization).
The receiver is given a channel output $Y=(Y_1,\dots,Y_m)$, that 
is usually modeled as a linear function of the input, plus 
gaussian noise $Y = \H X+W$, where $W=(W_1,\dots,W_m)$ are 
normal iid random variables.
Examples are CDMA or multiple-input multiple-output channels 
(with perfect channel state information)
\cite{VerduBook,TseBook}. 

The analysis simplifies considerably if the $X_i$'s are assumed to be 
normal as well \cite{TseHanly,VerduShamaiCDMA}. 
However, in many circumstances a binary
or quaternary modulation is used, and the normal assumption is therefore 
unrealistic.
The non-rigorous `replica method' from statistical physics have been used
to compute the channel capacity in these cases
\cite{TanakaCDMA}. A proof of replicas
predictions have been obtained in \cite{MonTse} in under some condition on the 
spreading factor. The same techniques were applied in more general settings
in \cite{GuoVerdu,GuoWang1,GuoWang2}.

However, specific assumptions on the spreading factor (and noise parameters) 
were necessary. Such assumptions ensured an appropriate density 
evolution operator to have unique fixed point.
The results of the present paper should allow to
prove replica results without conditions on the spreading.
\vspace{0.25cm}

As mentioned above we shall make a few technical assumptions 
on the structure of the sparse observation system.
These will concern the distribution of the bipartite graph connecting 
hidden variables and observations, 
as well as the dependency of the noisy observations on the $X$'s.
While such assumptions rule out some of the example above (for
instance, they exclude general irregular LDPC ensembles), we do not think they 
are crucial for the results to hold. 
%
%
\subsection{An informal overview}

We consider two types of observations: single variable observations
$Z=(Z_1,\dots, Z_n)$, and multi-variable observations 
$Y=(Y_1,\dots,Y_m)$.
For each $i\in [n]$, $Z_i$ is the result of observing $X_i$
through a memoryless noisy channel. Further for each $a$,
$Y_a$ is an independent noisy function of a subset
$\{X_j:\, j\in\da\}$ of the hidden variables. 
By this we mean that $Y_a$ is conditionally independent from all
the other variables, given $\{X_j:\, j\in\da\}$.
The subset $\da \subseteq [n]$ is itself random with, for each $i\in[n]$,
$i\in\da$ independently with probability $\gamma/n$.
\vspace{0.25cm}

Generalizing the above, we consider the conditional distribution of $X_i$, 
given $Y$ \emph{and} $Z$:
\begin{eqnarray}
\mu_i(x_i) \equiv \prob\{X_i=x_i|Y,Z\}\, .\label{eq:Marginal}
\end{eqnarray}
One may wonder whether additional information can be extracted by 
considering the correlation among hidden variables. 
Our first result is that for a generic subset of the variables,
these correlation vanish. This is stated informally below:
\begin{quote}
\emph{
For any uniformly random set of variable indices $i(1),\dots, i(k)\in [n]$
and any $\xi_1,\dots ,\xi_k\in \cX$
\begin{eqnarray}
\prob\{X_{i(1)}=\xi_1,\dots,X_{i(k)}=\xi_k|Y,Z\}\approx
\prob\{X_{i(1)}=\xi_1|Y,Z\}\cdots\prob\{X_{i(k)}=\xi_k|Y,Z\}\, .
\end{eqnarray}}
\end{quote}
This can be regarded as a generalization of the `decoupling 
principle' postulated in \cite{GuoVerdu}.
Here the $\approx$ symbols hides the large system ($n,m\to\infty$)
limit, and a `smoothing procedure' to be discussed below.
\vspace{0.25cm}

Locally, the graph $G$ converges to a random bipartite tree.
The locally tree-like structure of $G$ suggests the use of message passing 
algorithms, in particular belief propagation, 
for estimating the marginals $\mu_i$.
Consider the subgraph including $i$ as well as all the function nodes
$a$ such that $Y_a$ depends on $i$, and the other variables 
these observations depend on. Refer to the latter as to the 
`neighbors of $i$.' In belief propagation one assumes these to be independent 
\emph{in absence of $i$, and of its neighborhood}.

For any $j$, neighbor of $i$, let $\mu_{j\to i}$ denote the conditional
distribution of $X_j$ in the modified graph where $i$ (and the neighboring 
observations) have been taken out. Then BP provides a 
prescription\footnote{The mapping $\Fbp_i(\,\cdot\,)$ returns the 
marginal at $i$ with respect to the subgraph induced by $i$ and
its neighbors, when the latter are biased according to $\mu_{j\to i}$. 
For a more detailed description, we refer to Section \ref{sec:BP}.}
for computing $\mu_i$ in terms of the `messages' $\mu_{j\to i}$,
of the form $\mu_i=\Fbp_i(\{\mu_{j\to i}\})$.
We shall prove that this prescription is asymptotically correct.
\begin{quote}\emph{
Let $i$ be a uniformly random variable node and
$i(1),\dots, i(k)$ its neighbors. Then
\begin{eqnarray}
\mu_i \approx \Fbp_i(\mu_{i(1)\to i},\dots,\mu_{i(k)\to i})\, .
\label{eq:BPSymbolical}
\end{eqnarray}
}
\end{quote}

The neighborhood of $i$ converges to a Galton-Watson tree, with Poisson 
distributed degrees of mean $\gamma\alpha$ (for variable nodes, 
corresponding to variables $X_i$'s) 
and $\gamma$ (for function nodes corresponding to observation $Y_a$'s).
Such a tree is generated as follows. 
Start from a root variable node, generate a Poisson$(\gamma\alpha)$
number of function node descendants, and for each of
them an independent Poisson$(\gamma)$ number of variable node descendants.
This procedure is then repeated recursively.

In such a situation, consider again the BP equation (\ref{eq:BPSymbolical}).
Then the function $\Fbp_i(\,\cdots\,)$ can be approximated by a random function
corresponding to a random Galton-Watson neighborhood, do be denoted as $\Fde$. 
Further, 
one can hope that, if the graph $G$ is random, then the $\mu_{i(j)\to i}$
become iid random variables. Finally (and this is a specific
property of Poisson degree distributions) the residual graph with the 
neighborhood of $i$ taken out, has the same distribution (with slightly 
modified parameters) as the original one.
Therefore, one might imagine that the distribution of $\mu_i$
is the same as the one of the $\mu_{i(j)\to i}$'s.
Summarizing these observations one is lead to think that 
the distribution of $\mu_i$ must be (asymptotically for large systems)
a fixed point of the following distributional equation
\begin{eqnarray}
\nu \ed \Fde(\nu_1,\dots,\nu_l)\, .
\end{eqnarray}
This is an equation for the distribution of $\nu$ (the latter taking 
values in the set of distributions over the hidden variables $X_i$)
and is read as follows.
When $\nu_1,\dots,\nu_l$ are random variables with common 
distribution $\rho$, then $\Fde(\nu_1,\dots,\nu_l)$
has itself distribution $\rho$ (here $l$ and $\Fde$ are also
random according to the Galton-Watson model for the neighborhood of $i$).
It is nothing but the fixed point equation for density evolution, and
can be written more explicitly as
\begin{eqnarray}
\rho(\nu\in A) =\int \ind\big( \Fde(\nu_1,\dots,\nu_l)\in A\big)
\,\rho(\de\nu_1)\cdots\,\rho(\de\nu_l)\, ,
\end{eqnarray}
where $\ind(\,\cdots\,)$ is the indicator function.
In fact our main result tells that: $(i)$ The distribution
of $\mu_i$ must be a convex combination of the solutions of the above 
distributional equation; $(ii)$ If such convex combination is 
nontrivial (has positive weight on more than one solution) then
the correlations among the $\mu_i$'s have a peculiar structure.
\begin{quote}\emph{
Assume density evolution to admit the fixed point distributions
$\rho_{1},\dots, \rho_{r}$ for some fixed $r$. Then there 
exists probabilities $w_1,\dots,w_r$ (which add up to $1$) 
such that, for $i(1)\dots i(k)\in[n]$ uniformly random variable nodes,
\begin{eqnarray}
\prob\{\mu_{i(1)}\in A_1, \dots \mu_{i(k)}\in A_k\}
\approx \sum_{\alpha=1}^{r} w_{\alpha}\, 
\rho_{\alpha}(\mu\in A_1)\cdots \rho_{\alpha}(\mu\in A_k)\, .
\end{eqnarray}
}\end{quote}
In the last statement we have hidden one more technicality:
the stated asymptotic behavior might hold only along 
a subsequence of system sizes.
In fact in many cases it can be proved that the above convex 
combination is trivial, and that no subsequence needs to be taken.
Tools for proving this will be developed in a forthcoming publication.
%
%
\section{Definitions and main results}
\label{sec:Def}

In this section we provide formal definitions and statements.

\subsection{Sparse systems of observations}

We consider systems defined on a bipartite graph 
$G = (V,F,E)$, whereby $V$ and $F$ are vertices corresponding
(respectively) to variables and observations
(`variable' and `function nodes').
The edge set is $E\subseteq V\times F$.
For greater clarity, we shall use $i,j,k,\dots\in V$ to denote 
variable nodes and $a,b,c,\dots\in F$ for function nodes.
For $i\in V$, we let  $\di\equiv \{a\in F:\; (i,a)\in E\}$ denote
its neighborhood (and define analogously $\da$ for $a\in F$).
Further, if we let $n\equiv|V|$, and $m\equiv|F|$, we are interested
in the limit $n,m\to\infty$ with $\alpha = m/n$ kept fixed
(often we will identify $V = [n]$ and $F = [m]$).

A family of iid random variables $\{X_i:\, i\in V\}$,
taking values in a finite alphabet $\cX$, is associated with the 
vertices of $V$.
The common distribution of the $X_i$ will be denoted by 
$\prob\{X_i = x\}=p(x)$.
Given $U\subseteq V$, we let 
$X_U \equiv \{X_i:\, i\in U\}$ (the analogous convention is adopted
for other families of variables). Often we shall write $X$ for $X_V$.

Random variables $\{Y_a:\, a\in F\}$ are associated with the function nodes,
with $Y_a$ conditionally independent of $Y_{F\setminus a}$, 
$X_{V\setminus \da}$, given $X_{\da}$. Their joint distribution is 
defined by a set of probability kernels $Q^{(k)}$ indexed by
$k\in\naturals$, whereby, for $|\da| = k$, 
\begin{eqnarray}
\prob\{Y_a\in \,\cdot\, |X_{\da} = x_{\da}\}=
Q^{(k)}(\,\cdot\,|x_{\da})\, .
\end{eqnarray}
We shall assume $Q^{(k)}(\,\cdot\,|x_1,\dots,x_k)$ to be invariant 
under permutation of its arguments $x_1,\dots,x_k$
(an assumption that is implicit in the above equation).
Further, whenever clear from the context, we shall drop the 
superscript $(k)$. Without loss of generality, one can assume $Y_a$ 
to take values in $\reals^{b}$ for some $b$ which only depends on $k$.

A second collection of real random 
variables $\{Z_i:\; i\in V\}$ is associated with
the variable nodes, with $Z_i$ conditionally independent of
$Z_{V\setminus i}$, $X_{V\setminus i}$ and $Y$, conditional on
$X_i$. The associated probability
kernel will be denoted by $R$:
\begin{eqnarray}
\prob\{Z_i\in\,\cdot\, |X_i=x_i\} = R(\,\cdot\,|x_i)\, .
\end{eqnarray}

Finally, the graph $G$ itself will be random. All the above distributions 
have to be interpreted as conditional to a given realization of $G$.
We shall follow the convention of using $\prob\{\,\cdots\,\}$,
$\E\{\,\cdots\,\}$ etc, for conditional probability, expectation, etc.
given $G$ (without writing explicitly the conditioning) and write 
$\prob_G\{\,\cdots\,\}$, $\E_G\{\,\cdots\,\}$ for probability and expectation
with respect to $G$. The graph distribution is defined as follows. 
Both node sets $V$ and $F$ are given. Further, for any $(i,a)\in V\times F$,
we let $(i,a)\in E$ independently with probability $p_{\rm edge}$.
If we let $n\equiv|V|$, and $m\equiv|F|$
(often identifying $V = [n]$ and $F = [m]$),
such a random graph ensemble will be denoted as 
$\graph(n,m,p_{\rm edge})$.
We are interested
in the limit $n,m\to\infty$ with $\alpha = m/n$ kept fixed 
and $p_{\rm edge}=\gamma/n$.

In particular, we will be concerned with the problem of
determining the conditional distribution of $X_i$ given $Y$ and $Z$
cf. Eq.~(\ref{eq:Marginal}).
Notice that $\mu_i$ is a random variable taking values
in $\meas(\cX)$ (the set of probability measures over $\cX$).

In order to establish our main result we need to `perturb' the system
as follows. Given a perturbation parameter $\theta\in[0,1]$
(that should be thought as `small'), and a symbol $\ast\not\in\cX$, we let
\begin{eqnarray}
Z_i(\theta) = \left\{
\begin{array}{ll}
(Z_i,X_i)  &\mbox{ with probability $\theta$,}\\
(Z_i,\ast)  &\mbox{ with probability $1-\theta$.} 
\end{array}
\right.
\end{eqnarray}
In words, we reveal a random subset of the hidden variables. 
Obviously $Z(0)$ is equivalent to  $Z$ and $Z(1)$ to  $X$.
The corresponding probability kernel is defined by
(for $A\subseteq \reals$ measurable, and $\overline{x}\in\cX\cup\{\ast\}$) 
\begin{eqnarray}
R^{\theta}(\overline{x},A|x_i) = [(1-\theta)\ind(\overline{x} = \ast)+
\theta\, \ind(\overline{x}=x_i)]
R(A|x_i)\, ,
\end{eqnarray}
where $\ind(\,\cdots\,)$ is the indicator function.
We will denote by $\mu^{\theta}_i$ the analogous of $\mu_i$, 
cf. Eq.~(\ref{eq:Marginal}), with $Z$ being replaced by $Z(\theta)$.

It turns out that introducing such a perturbation is necessary 
for our result to hold. 
The reason is that there can be specific choices of the 
system `parameters' $\alpha$, $\gamma$, and of the kernels
$Q$ and $R$ for which the variables $X_i$'s are strongly correlated.
This happens for instance at thresholds noise levels in coding. 
Introducing a perturbation allows to remove this non-generic
behaviors.

We finally need to introduce a technical regularity condition 
on the laws of $Y_a$ and $Z_a$ (notice that this concerns the
\emph{unperturbed model}). 
\begin{definition}
We say that a probability kernel $T$ from $\cX$ to a measurable 
space $\cS$
(i.e., a set of probability measures $T(\,\cdot\,|x)$
indexed by $x\in\cX$) is \emph{soft} if: 
$(i)$  $T(\,\cdot\,|x_1)$ is absolutely continuous with respect to 
 $T(\,\cdot\,|x_2)$ for any $x_1,x_2\in\cX$; $(ii)$
We have, for some $M<\infty $, and all $x\in\cX$
(the derivative being in the Radon-Nikodyn sense)
\begin{eqnarray}
\int \frac{\de T(y|x_1)}{\de T(y|x_2)}\; T(\de y|x)\le M\, .
\end{eqnarray}

A system of observations is said to have 
\emph{soft noise} (or \emph{soft noisy observations}), 
if there exists $M<\infty$ such that 
the kernels $R$ and and $Q^{(k)}$ for all $k\ge 1$ are $M$-soft.
\end{definition}
In the case of a finite output alphabet  the above definition simplifies
considerably: a kernel is soft if all its entries are non-vanishing.
Although there exist interesting examples of non-soft kernels
(see, for instance, Section \ref{sec:Examples})
they can often be treated as limit cases of soft ones.
%
%
\subsection{Belief propagation and density evolution}
\label{sec:BP}

Belief propagation (BP) is frequently used in practice  
to estimate the marginals (\ref{eq:Marginal}). Messages 
$\me^{(t)}_{i\to a}$, $\mh^{(t)}_{a\to i}\in\meas(\cX)$ are exchanged at time 
$t$ along edge $(i,a)\in E$, where $i\in V$, $a\in F$.
The update rules follow straightforwardly from the general factor
graph formalism \cite{KFL01}
\begin{eqnarray}
\me^{(t+1)}_{i\to a}(x_i) & \propto &
p(x_i)R^{\theta}(z_i|x_i)\prod_{b\in\di\setminus a}
\mh^{(t)}_{b\to i}(x_i)\, ,\label{eq:BP1}\\
\mh^{(t)}_{a\to i}(x_i) & \propto & \sum_{x_{\da\setminus i}} Q(y_a|x_{\da})
\prod_{j\in \da\setminus i} \me^{(t)}_{j\to a}(x_j)\, .
\label{eq:BP2}
\end{eqnarray}
Here and below we denote by $\propto$ equality among measures
on the same space `up to a normalization\footnote{Explicitly,
$q_1(x)\propto q_2(x)$ if there exists a constant $C> 0$
such that $q_1(x) = q_2(x)$ for all $x$}.'
The BP estimate for the marginal of variable $X_i$ is (after $t$
iterations)
\begin{eqnarray}
\me^{(t+1)}_{i}(x_i) & \propto &
p(x_i)R^{\theta}(z_i|x_i)\prod_{b\in\di}
\mh^{(t)}_{b\to i}(x_i)\, .\label{BPMarginal}
\end{eqnarray}

Combining Eqs.~(\ref{eq:BP1}) and (\ref{BPMarginal}),
the BP marginal at variable node $i$ can be expressed as a function 
of variable-to-function node messages at neighboring
variable nodes. We shall write
\begin{eqnarray}
\me^{(t+1)}_{i} &= &\Fbp_i(\{\me^{(t)}_{j\to b}:\, 
j\in\partial b\setminus i; \, b\in \di\})\, ,\label{eq:BPMarginalNew}\\
\Fbp_i(\cdots)(x_i) & \propto &
p(x_i)R^{\theta}(z_i|x_i)\prod_{a\in\di}
\left\{\sum_{x_{\da\setminus i}} Q(y_a|x_{\da})
\prod_{j\in \da\setminus i} \me^{(t)}_{j\to a}(x_j)\right\}\, .
\end{eqnarray}

Notice that the mapping $\Fbp_i(\,\cdots\,)$ only depends on the graph
$G$ and on the observations $Y,Z(\theta)$, through the subgraph
including function nodes adjacent to $i$ and the corresponding variable nodes.
Denoting such neighborhood as $\Ball$, the corresponding observations
as $Y_{\Ball}$, $Z_\Ball(\theta)$, and letting 
$\me^{(t)}_{\dBall}=\{\me^{(t)}_{j\to b}:\,j\in\partial b\setminus i; 
\, b\in \di\}$, we can rewrite Eq.~(\ref{eq:BPMarginalNew}) in the form
\begin{eqnarray}
\me^{(t+1)}_{i} &= &\Fbp(\me^{(t)}_{\dBall}\, ;\Ball,Y_{\Ball}, 
Z_{\Ball}(\theta) ) \, .
\end{eqnarray}
Here we made explicit all the dependence upon the graph and the observations.
If $G$ is drawn randomly from the $\graph(n,\alpha n,\gamma/n)$
ensemble, the neighborhood $\Ball$, as well as the corresponding
observations converge in the large system limit,
to a well defined limit distribution. 
Further, the messages $\{\me^{(t)}_{j\to b}\}$ above become iid and are 
distributed as $\mu_i^{(t)}$ (this is a consequence of the fact that the 
edge degrees are asymptotically Poisson). Their common distribution satisfies 
the density evolution distributional recursion
\begin{eqnarray}
\me^{(t+1)} &\ed &\Fde(\me^{(t)}_{\dBall}\, ;\Ball,Y_{\Ball}, 
Z_{\Ball}(\theta) ) \, ,\label{eq:DensityEvolution}
\end{eqnarray}
where $\me^{(t+1)}_{\dBall} = \{\me^{(t)}_e:\, e\in \dBall\}$ are iid copies
of $\me^{(t)}$, and $\Ball,Y_{\Ball}, Z_{\Ball}(\theta)$ are understood 
to be taken from their asymptotic distribution.
We will be particularly concerned with the set of \emph{fixed points}
of the above distributional recursion. This is just the set of distributions
$\rho$ over $\meas(\cX)$ such that, if $\nu^{(t)}$ has distribution $\rho$,
then $\nu^{(t+1)}$ has distribution $\rho$ as well.
%
%
\subsection{Main results}

For stating our first result, 
it is convenient to introduce a shorthand notation.
For any $U\subseteq V$, we note
\begin{eqnarray}
\tprob_U\{x_U\} \equiv \prob\{X_U=x_U|Y,Z(\theta)\}\, .
\end{eqnarray}
Notice that, being a function of $Y$ and $Z(\theta)$, 
$\tprob_U\{x_U\}$ is a random variable. 
The theorem below shows that, if $U$ is a random subset 
of $V$ of bounded size, then $\tprob_U$ factorizes approximately over 
the nodes $i\in U$. The accuracy of this is measured in terms of 
\emph{total variation distance}. Recall that, given two distributions
$q_1$ and $q_2$ on the same finite set $\cS$, their total variation distance is
\begin{eqnarray}
||q_1-q_2||_{\sTV} = \frac{1}{2}\sum_{x\in\cS}|q_1(x)-q_2(x)|\, .
\end{eqnarray}
\begin{thm}\label{thm:Correlation}
Consider an observation system on the  graph $G=(V,F,E)$. Let 
$k\in \naturals$, and $i(1),\dots, i(k)$ be uniformly random in
$V$. Then, for any $\eps>0$  
\begin{eqnarray}
\int_0^{\eps}\E_{i(1)\cdots i(k)}
\E\Big|\Big|\tprob_{i(1),\dots,i(k)}-\tprob_{i(1)}\cdots\tprob_{i(k)}
\Big|\Big|_{\sTV}\;\de\theta\le (|\cX|+1)^kA_{n,k}\sqrt{H(X_1)\eps/n}
= O(n^{-1/2})\, ,
\label{eq:FactorizationThm}
\end{eqnarray}
where $A_{n,k}\le \exp\left(\frac{k^2}{2n}\right)$ for $k<n/2$,
and the asymptotic behavior $O(n^{-1/2})$ holds as $n\to\infty$ with
if $k$ and $\cX$ fixed.
\end{thm}

The next result establishes that the BP equation (\ref{eq:BPMarginalNew})
is approximately satisfied by the actual marginals.
For any $i$, $j\in V$, such that $i,j\in\partial b$ for some common 
function node $b\in F$, let
\begin{eqnarray}
\mu^{\theta(j)}_i(x_i)\equiv\prob\Big\{X_i=x_i\Big| 
Y_a: j\not\in\da;
\;Z_l(\theta): l \neq j\Big\}\, .
\end{eqnarray}
This is nothing but the conditional distribution of $X_i$ with
respect to the graph from which $j$ has been `taken out.'
\begin{thm}\label{thm:ApproximateMarginal}
Consider a sparse observation system on a random graph $G = (V,F,E)$ from
the $\graph_{n}(\gamma/n,\alpha n)$ ensemble. Assume the
noisy observations to be $M$-soft. Then there exists a constant 
${\sf A}$ depending on $t,\alpha,\gamma,M,|\cX|,\eps$, such that
for any $i\in V$, and any $n$
\begin{eqnarray}
\int_0^{\eps}\E_G\E||\mu_i^{\theta}-
\Fbp_i(\{\mu_{j}^{\theta,(i)}\}_{a\in\di, j\in\da
\setminus i})||_{\sTV}
\;\de\theta \le \frac{{\sf A}}{\sqrt{n}}\, .
\end{eqnarray}
\end{thm}

Finally, we provide a characterization of the asymptotic distribution of 
the one variable marginals.
Recall that $\meas(\cX)$ denotes the set of probability distributions
over $\cX$, i.e., the $(|\cX|-1)$-dimensional standard  simplex. We further let
$\meas^2(\cX)$ be the set of probability measures over  
$\meas(\cX)$ ($\meas(\cX)$ being endowed with the Borel 
$\sigma$-field induced by 
$\reals^{|\cX|-1}$). This can be  equipped with the 
smallest $\sigma$-field that makes $F_A:\rho\mapsto \rho(A)$ measurable for any
Borel subset $A$ of $\meas(\cX)$. 
\begin{thm}\label{thm:Main}
Consider an observation system on a random graph $G = (V,F,E)$ from
the $\graph(n,\alpha n,\gamma/n)$ ensemble, and assume the
noisy observations to be soft.
Let $\varphi:\meas(\cX)^k\to\reals$ be a Lipschitz continuous
function on $\meas(\cX)^k=\meas(\cX)\times\cdots\times\meas(\cX)$
($k$ times).

Then for almost any $\theta\in [0,\ve]$ there exists 
an infinite subsequence $R_{\theta}\subseteq \naturals$ and a 
probability distribution $S_{\theta}$ over $\meas^2(\cX)$, supported on
the fixed points of the density evolution recursion 
(\ref{eq:DensityEvolution}), such that the following happens.
Given any fixed subset of variable nodes  $\{i(1),\dots,i(k)\}\subseteq V$
\begin{eqnarray}
\lim_{n\in R_{\theta}}\,  \E_G\E\; \big\{\varphi(\mu_{i(1)}^{\theta},\dots,
\mu_{i(k)}^{\theta})\big\}= \int \left\{\int
\varphi(\mu_1,\dots,\mu_k)\, \rho(\de\mu_1)\cdots\rho(\de\mu_k)
\right\}S(\de\rho)\, .\label{eq:Main}
\end{eqnarray}
\end{thm}
%
%
\section{Proof of Theorem \ref{thm:Correlation} (correlations)}
\label{sec:Correlation}

\begin{lemma}\label{lemma:TwoPoints}
For any observation system and any $\eps>0$ 
\begin{eqnarray}
\frac{1}{n}\sum_{i,j\in V}\int_{0}^{\eps}\!I(X_i;X_j|Y,Z(\theta))\,\de\theta\le 2 H(X_1)\, .\label{eq:SumMutualInfo}
\end{eqnarray}
\end{lemma}
\begin{proof}
For $U\subseteq V$, let us denote by $Z^{(U)}(\theta)$
the vector obtained by setting $Z^{(U)}_i(\theta) =Z_i(\theta)$
whenever $i\not\in U$, and  $Z^{(U)}_i(\theta) =(Z_i,\ast)$ if $i\in U$.
The proof is based on the two identities below
\begin{eqnarray}
\frac{\de\phantom{\theta}}{\de\theta} H(X|Y,Z(\theta)) & =& - 
\sum_{i\in V} H(X_{i}|Y, Z^{(i)}(\theta))\, ,\label{eq:FirstDerivative}\\
\frac{\de^2\phantom{\theta}}{\de\theta^2} H(X|Y,Z(\theta)) & =& 
\sum_{i\neq j\in V} I(X_{i};X_{j}|Y, Z^{(ij)}(\theta))\, .
\label{eq:SecondDerivative}
\end{eqnarray}
Before proving these identities, let us show that they imply the thesis.
By the fundamental theorem of calculus, we have
\begin{eqnarray}
\frac{1}{n}\sum_{i\neq
j\in V}\int_{0}^{\eps}\!I(X_i;X_j|Y,Z^{(ij)}(\theta))\,\de\theta
&=&\frac{1}{n}\sum_{i\in V} H(X_{i}|Y, Z^{(i)}(0))-
\frac{1}{n}\sum_{i\in V} H(X_{i}|Y, Z^{(i)}(\eps))\\
&\le &\frac{1}{n}\sum_{i\in V} H(X_{i}|Y, Z^{(i)}(0))\le H(X_1)\, .
\end{eqnarray}
Further, if $z^{(U)}(\theta)$ is the vector obtained
from $z(\theta)$ by replacing $z_i(\theta)$ with $(z_i,\ast)$ for
any $i\in U$, then
\begin{eqnarray}
I(X_i;X_j|Y,Z(\theta) = z(\theta)) \le 
I(X_i;X_j|Y,Z^{(ij)}(\theta) = z^{(ij)}(\theta)) \, .
\end{eqnarray}
In fact the left hand side vanishes whenever $z^{(ij)}(\theta)\neq
z(\theta)$.
The proof is completed by upper bounding the diagonal terms in the sum
(\ref{eq:SumMutualInfo}) as $I(X_i;X_i|Y,Z^{(i)}(\theta))=
H(X_i|Y,Z^{(i)}(\theta))\le H(X_1)$.

Let us now consider the identities (\ref{eq:FirstDerivative})
and (\ref{eq:SecondDerivative}). These
already appeared in the literature \cite{MMU05,MMRU05,MacrisGKS}. We reproduce 
the proof here for the sake of self-containedness.

Let us begin with Eq.~(\ref{eq:FirstDerivative}). 
It is convenient to slightly generalize the model by letting 
the parameter the channel parameter $\theta$ be dependent on the variable node.
In other words given a vector $\uth = (\theta_1,\dots,
\theta_n)$, we let, for each $i\in V$, $Z_i(\uth) = (Z_i,X_i)$ with 
probability $\theta_i$, and $=(Z_i,\ast)$ otherwise.
Noticing that  $H(X|Y,Z(\uth))= 
H(X_i|Y,Z(\uth))+H(X|X_i,Y,Z(\uth))$ and that the latter
term does not depend upon $\theta_i$, we have
\begin{eqnarray}
\frac{\partial\phantom{\theta_i}}{\partial\theta_i} H(X|Y,Z(\uth))
=\frac{\partial\phantom{\theta_i}}{\partial\theta_i} 
H(X_i|Y,Z(\uth)) = -H(X_i|Y,Z^{(i)}(\uth)) \, ,
\end{eqnarray}
where the second equality is a consequence of
$H(X_i|Y,Z(\uth)) = (1-\theta_i)H(X_i|Y,Z^{(i)}(\uth))$.
Equation (\ref{eq:FirstDerivative}) follows by simple calculus taking 
$\theta_i=\theta_i(\theta) = \theta$ for all $i\in V$.
 
Equation (\ref{eq:SecondDerivative}) is proved analogously.
First, the above calculation implies that the second derivative with respect to 
$\theta_i$ vanishes for any $i\in V$. 
For $i\neq j$, we 
use the chain rule to get $H(X|Y,Z(\uth)) = H(X_i,X_j|Y,Z(\uth))+
H(X|X_i,X_j,Y,Z^{(ij)}(\uth))$, and then write
\begin{eqnarray*}
H(X_i,X_j|Y,Z(\uth)) = (1-\theta_i)(1-\theta_j)H(X_i,X_j|Y,Z^{(ij)}(\uth))
+\theta_i(1-\theta_j)H(X_j|X_i,Y,Z^{(ij)}(\uth))+\\
+(1-\theta_i)\theta_jH(X_i|X_j,Y,Z^{(ij)}(\uth))\, ,
\end{eqnarray*}
whence the mixed derivative with respect to $\theta_i$ and $\theta_j$
results in $I(X_i;X_j|Y,Z^{(ij)}(\uth))$.
As above, Eq.~(\ref{eq:SecondDerivative}) is recovered by letting 
$\theta_i=\theta_i(\theta) =\theta$ for any $i\in V$.
\end{proof}

In the next proof we will use a technical device that 
has been developed within the mathematical theory of spin glasses
(see \cite{Talagrand}, and 
\cite{GuerraToninelli,GerschenMontanari} for applications to sparse models).
We start by defining a family of real random variables indexed
by a variable node $i\in V$, and by $\xi\in \cX$:
\begin{eqnarray}
\Sp_i(\xi) \equiv \ind(X_i=\xi)-\prob\{X_i=\xi|Y,Z(\theta)\}\, .
\end{eqnarray}
We will also use $\Sp(\xi) = (\Sp_1(\xi)\,\dots,\Sp_n(\xi))$
to denote the corresponding vector.

Next we let $X^{(1)} = (X^{(1)}_1,\dots,X^{(1)}_n)$ and 
$X^{(2)} = (X^{(2)}_1,\dots,X^{(2)}_n)$ be two iid assignments of the
hidden variables, both distributed according to the conditional law
$\prob_{X|Y,Z(\theta)}$. If we let $(Y,Z(\theta))$ be distributed according
to the original (unconditional) law $\prob_{Y,Z(\theta)}$, this defines
a larger probability space, generated by $(X^{(1)},X^{(2)},Y,Z)$.
Notice that the pair  $(X^{(1)},Y,Z)$ and  $(X^{(2)},Y,Z)$ is
exchangeable, each of the terms being distributed as $(X,Y,Z(\theta))$.

In terms of $X^{(1)}$ and $X^{(2)}$ we can then define 
$\Sp^{(1)}(\xi)$ and $\Sp^{(2)}(\xi)$, and introduce the \emph{overlap}
\begin{eqnarray}
\Qp(\xi)  \equiv \frac{1}{n}\, \Sp^{(1)}(\xi)\cdot\Sp^{(1)}(\xi)
=\frac{1}{n}\sum_{i\in V} \Sp^{(1)}_i(\xi)\, \Sp^{(2)}_i(\xi)\, .
\end{eqnarray}
Since $|\Sp_i(\xi)|\le 1$, we have $|\Qp(\xi)|\le 1$ as well.
Our next result shows that the conditional distribution
of $\Qp(\xi)$ given $Y$ and $Z(\theta)$ is indeed very concentrated, 
for most valued of $\theta$. The result is expressed in terms of
the conditional variance
\begin{eqnarray}
\Var(\Qp(\xi) |Y,Z(\theta))\equiv
\E\left\{\E[\Qp(\xi)^2|Y,Z(\theta)]-\E[\Qp(\xi)|Y,Z(\theta)]^2\right\}\, .
\end{eqnarray}

\begin{lemma}\label{lemma:VarianceBound}
For any observations system  and any $\eps>0$ 
\begin{eqnarray}
\int_0^{\eps}\!\Var(\Qp(\xi) |Y,Z(\theta))\;\de\theta \le 4H(X_1)/n\, .
\end{eqnarray}
\end{lemma}
\begin{proof}
In order to lighten the notation, write $\tE\{\,\cdots\,\}$ 
for $\E\{\,\cdot\,|Y=y,Z(\theta) = z(\theta)\}$ (and
analogously for $\tprob\{\,\cdots\,\}$), and drop the argument $\xi$
from $\Sp^{(a)}_i(\xi)$. Then
\begin{eqnarray*}
\Var(\Qp(\xi) |Y=y,Z(\theta)=z(\theta)) & = & 
\tE\left\{\left(\frac{1}{n}\sum_{i\in V}\Sp_i^{(1)}\Sp_i^{(2)}
\right)^2\right\}-
\tE\left\{\frac{1}{n}\sum_{i\in V}\Sp_i^{(1)}\Sp_i^{(2)}
\right\}^2 =\\
& = & 
\frac{1}{n^2}\sum_{i,j\in V}\left\{
\tE\left\{\Sp_i^{(1)}\Sp_i^{(2)}\Sp_j^{(1)}\Sp_j^{(2)}\right\}-
\tE\left\{\Sp_i^{(1)}\Sp_i^{(2)}\right\}
\tE\left\{\Sp_j^{(1)}\Sp_j^{(2)}\right\}\right\} =\\
& = & 
\frac{1}{n^2}\sum_{i,j\in V}\left\{
\tE\left\{\Sp_i\Sp_j\right\}^2-
\tE\left\{\Sp_i\right\}^2
\tE\left\{\Sp_j\right\}^2\right\}\, .
\end{eqnarray*}
In the last step we used the fact that $\Sp^{(1)}(\xi)$ and $\Sp^{(2)}(\xi)$ 
are conditionally independent given $Y$ and $Z(\theta)$, and used 
the notation $\Sp_i(\xi)$ for any of them (recall that 
$\Sp^{(1)}(\xi)$ and $\Sp^{(2)}(\xi)$ are identically distributed).
Notice that
\begin{eqnarray}
\tE\{\Sp_i(\xi)\} &=& 
\E\Big\{\ind(X_i=\xi)-\prob\{X_i=\xi|Y,Z(\theta)\}\Big|
Y=y,Z(\theta) =z(\theta)\Big\} = 0\, ,\\
\tE\{\Sp_i(\xi)\Sp_j(\xi)\} &=&
 \tE\Big\{\big[\ind(X_i=\xi)-\prob\{X_i=\xi|Y,Z(\theta)\}\big]
\big[\ind(X_j=\xi)-\prob\{X_j=\xi|Y,Z(\theta)\big]\Big\} =\nonumber\\
&=&\tprob\{X_i=\xi,\,X_j=\xi\}-\tprob\{X_i=\xi\}\tprob\{X_j=\xi\}  \, .
\end{eqnarray} 
Therefore 
\begin{eqnarray*}
\Var(\Qp(\xi) |Y=y,Z(\theta)=z(\theta)) & = &
\frac{1}{n^2}\sum_{i,j\in V}\Big(\tprob\{X_i=\xi,\,X_j=\xi\}-\tprob\{X_i=\xi\}\tprob\{X_j=\xi\}\Big)^2 \le\\
&\le& \frac{1}{n^2}\sum_{i,j\in V}\sum_{x_1,x_2}
\Big(\tprob\{X_i=x_1,\,X_j=x_2\}-\tprob\{X_i=x_1\}\tprob\{X_j=x_2\}\Big)^2\le\\
&\le& \frac{2}{n^2}\sum_{i,j\in V} I(X_i;X_j|Y=y,Z(\theta)=z(\theta)))\, .
\end{eqnarray*}
In the last step we used the inequality (valid for any two distributions 
$p_1$, $p_2$ over a finite set $\cS$)
\begin{eqnarray}
\sum_x\big|p_1(x)-p_2(x)\big|^2\le 2 D(p_1||p_2)\, ,
\end{eqnarray}
and applied it to the joint distribution of $X_1$ and $X_2$, and the
product of their marginals.
The thesis follows by integrating over $y$ and $z(\theta)$ with the measure 
$\prob_{Y,Z(\theta)}$ and using Lemma \ref{lemma:TwoPoints}.
\end{proof}
\begin{proof}[Proof (Theorem \ref{thm:Correlation}).]
We start by noticing that, since $|\Qp(\xi)|\le 1$, and 
$\tE\{\Qp(\xi)\}=0$, we have, for any $\xi_1,\dots,\xi_k\in\cX$, 
\begin{eqnarray*}
|\tE\{\Qp(\xi_1)\cdots \Qp(\xi_k)\}|&\le&
|\tE\{\Qp(\xi_1)\Qp(\xi_2)\}|\le \sqrt{
 \tE\{\Qp(\xi_1)^2\}\tE\{\Qp(\xi_2)^2\}}\le\\
&\le&\frac{1}{2}\Var(\Qp(\xi_1)|Y=y,Z=z(\theta))+
\frac{1}{2}\Var(\Qp(\xi_2)|Y=y,Z=z(\theta))\, ,
\end{eqnarray*}
(where we assumed, without loss of generality, $k\ge 2$).
Integrating with respect to $y$ and $z(\theta)$ with the measure 
$\prob_{Y,Z(\theta)}$, and using Lemma \ref{lemma:VarianceBound},
we obtain 
\begin{eqnarray}
\int_0^\eps\E\Big|\E\{\Qp(\xi_1)\cdots \Qp(\xi_k)|Y,Z(\theta)\}\Big|\,\de\theta
\le 4H(X_1)/n\, .\label{eq:KthMoment}
\end{eqnarray}
On the other hand
\begin{eqnarray}
\tE\{\Qp(\xi_1)\cdots \Qp(\xi_k)\} &=& \frac{1}{n^k}\sum_{j(1)\dots j(k)\in V}
\tE\{\Sp^{(1)}_{j(1)}(\xi_1)\Sp^{(2)}_{j(1)}(\xi_2)\cdots
\Sp^{(1)}_{j(k)}(\xi_k)\Sp^{(2)}_{j(k)}(\xi_k)
\}=\\
&=&\frac{1}{n^k}\sum_{j(1)\dots j(k)\in V}
\tE\{\Sp_{j(1)}(\xi_1)\cdots\Sp_{j(k)}(\xi_k)\}^2 \ge\\
& \ge &\frac{k!}{n^k}\binom{n}{k}
\E_{i(1)\dots i(k)} \tE\{\Sp_{i(1)}(\xi_1)\cdots\Sp_{i(k)}(\xi_k)\}^2\, .
\label{eq:Kpoints}
\end{eqnarray}
Putting together Eq.~(\ref{eq:KthMoment}) and (\ref{eq:Kpoints}),
letting $B_{n,k} \equiv n^k/k!\binom{n}{k}$, and taking expectation with 
respect to $Y$ and $Z(\theta)$, we get 
\begin{eqnarray}
\int_0^{\eps}\E_{i(1)\dots i(k)} 
\E\Big\{\E\{\Sp_{i(1)}(\xi_1)\cdots\Sp_{i(k)}(\xi_k)|Y,Z(\theta)\}^2\Big\}\,
\de\theta\le 4B_{n,k}H(X_1)/n\, ,
\end{eqnarray}
which, by Cauchy-Schwarz inequality, implies
\begin{eqnarray}
\int_0^{\eps}\E_{i(1)\dots i(k)} 
\E\Big\{\big|\E\{\Sp_{i(1)}(\xi_1)\cdots\Sp_{i(k)}(\xi_k)|Y,Z(\theta)\}
\big|\Big\}\,
\de\theta\le \sqrt{4\eps B_{n,k}H(X_1)/n}\, .\label{eq:IntKpoints}
\end{eqnarray}

Next notice that
\begin{align*}
\Big|\Big|\tprob_{i(1)\dots i(k)}-\tprob_{i(1)}\cdots\tprob_{i(k)}
\Big|\Big|_{\sTV} &= \frac{1}{2}
\sum_{\xi_1\dots\xi_k\in\cX}\Big|
\tprob_{i(1)\dots i(k)}\{\xi_1,\dots,\xi_k\}-
\tprob_{i(1)}\{\xi_1\}\cdots\tprob_{i(k)}\{\xi_k\}\Big|=\\
=&\frac{1}{2}
\sum_{\xi_1\dots\xi_k\in\cX}\Big|\tE\left\{
\ind(X_{i(1)}=\xi_1)\cdots \ind(X_{i(k)}=\xi_k)-\tprob_{i(1)}\{\xi_1\}\cdots\tprob_{i(k)}\{\xi_k\}
\right\}\Big|=\\
=&\frac{1}{2}
\sum_{\xi_1\dots\xi_k\in\cX}\left|\sum_{J\in [k],\, |J|\ge 2}
\tE\left\{\prod_{\alpha\in J}\Sp_{i(\alpha)}(\xi_{\alpha})\right\}
\prod_{\beta\in [k]\setminus J}\tprob_{i(\beta)}\{\xi_\beta\}\right|\, .
\end{align*}
Using triangular inequality 
\begin{align*}
\Big|\Big|\tprob_{i(1)\dots i(k)}-\tprob_{i(1)}\cdots\tprob_{i(k)}
\Big|\Big|_{\sTV} &\le 
\frac{1}{2}\sum_{J\in [k],\, |J|\ge 2}
\sum_{\{\xi_\alpha\}_{\alpha\in J}}\left|
\tE\left\{\prod_{\alpha\in J}\Sp_{i(\alpha)}(\xi_{\alpha})\right\}
\right|\, .
\end{align*}
Taking expectation with respect to $Y,Z(\theta)$ and
to $\{i(1),\dots,i(k)\}$ a uniformly random subset of 
$V$, we obtain
\begin{align*}
\E_{i(1)\dots i(k)} 
\E\Big|\Big|\tprob_{i(1)\dots i(k)}-\tprob_{i(1)}\cdots\tprob_{i(k)}
&\Big|\Big|_{\sTV}\le\\
&\le \frac{1}{2}\sum_{l=2}^k \binom{k}{l}
\sum_{\xi_1\dots\xi_l\in\cX}\E_{i(1)\dots i(l)} 
\E\Big|\E\{\Sp_{i(1)}(\xi_1)\cdots\Sp_{i(l)}(\xi_l)\big| Y,Z(\theta)\}\Big|\, .
\end{align*}
Integrating over $\theta$ and using Eq.~(\ref{eq:IntKpoints}), we get
\begin{eqnarray}
\int_0^{\eps}\E_{i(1)\cdots i(k)}
\E\Big|\Big|\tprob_{i(1),\dots,i(k)}-\tprob_{i(1)}\cdots\tprob_{i(k)}
\Big|\Big|_{\sTV}\;\de\theta\le \frac{1}{2}\sum_{l=2}^k
 \binom{k}{l}|\cX|^l\sqrt{4\eps B_{n,l}H(X_1)/n}\, .
\end{eqnarray}
By using $B_{n,l}\le B_{n,k}$ the right hand side is bounded as in 
Eq.~(\ref{eq:FactorizationThm}), with $A_{n,k}\equiv\sqrt{B_{n,k}}$.
The bound on this coefficient is obtained by a standard manipulation
(here we use $-\log(1-x)\le 2x$ for $x\in[0,1/2]$ and the hypothesis
$k\le n/2$):
\begin{eqnarray}
B_{n,k} = \exp\left\{-\sum_{i=1}^{k-1}\log\left(1-\frac{i}{n}\right)\right\}
\le\exp\left\{\sum_{i=1}^{k-1}\frac{2i}{n}\right\}= \exp\left\{
\frac{k(k-1)}{n}\right\}
\, ,
\end{eqnarray}
hence $A_{n,k}\le\exp\{k^2/2n\}$ as claimed. 
\end{proof}

Obviously, if the graph $G$ is `sufficiently' random,
the expectation over variable nodes $i(1),\dots,i(k)$
can be replaced by the expectation over $G$. 
\begin{coro}\label{coro:Correlation}
Let $G=(V,F,E)$ be a random bipartite graph whose distribution is
invariant under permutation of the variable nodes in $V=[n]$.
Then, for any observations system on $G=(V,F,E)$, any $k\in \naturals$
any $\eps>0$, and any (fixed) set of variable nodes 
$\{i(1),\dots,i(k)$,
\begin{eqnarray}
\int_0^{\eps}\E_{G}
\E\Big|\Big|\tprob_{i(1),\dots,i(k)}-\tprob_{i(1)}\cdots\tprob_{i(k)}
\Big|\Big|_{\sTV}\;\de\theta\le (|\cX|+1)^kA_{n,k}\sqrt{H(X_1)\eps/n}
= O(n^{-1/2})\, ,
\label{eq:FactorizationCoro}
\end{eqnarray}
where the constant $A_{k,n}$ is as in Theorem  \ref{thm:Correlation}.
\end{coro}

%
%
\section{Random graph properties}

The proofs of Theorems \ref{thm:ApproximateMarginal} and \ref{thm:Main} 
rely on some specific properties of the graph ensemble
$\graph(n,\alpha n,\gamma/n)$.

We begin with some further definitions concerning a generic
bipartite graph $G=(V,F,E)$.
Given $i,j\in V$, their graph-theoretic distance is defined 
as the length of the shortest path from $i$ to $j$ on $G$.
We follow the convention of measuring the length of a path 
on $G$ by the number of function nodes traversed by the path. 

Given $i\in V$ and $t\in\naturals$ we let $\Ball(i,t)$ be the 
subset of variable nodes $j$ whose distance from $i$ is at most 
$t$. With an abuse of notation, we use the same symbol to denote the 
subgraph induced by this set of vertices, i.e. the factor graph
including those function node $a$ such that $\da\subseteq\Ball(i,t)$
and all the edges incident on them.
Further, we denote by $\cBall(i,t)$ the subset of variable nodes
$j$ with $d(i,j)\ge t$, as well as the induced subgraph.
Finally $\dBall(i,t)$ is the subset of vertices with $d(i,j)=t$.
Equivalently $\dBall(i,t)$ is the intersection of $\Ball(i,t)$ and
$\cBall(i,t)$. 
\vspace{0.1cm}

We will make use of two remarkable properties of the ensemble
$\graph(n,n\alpha,\gamma/n)$:
$(i)$ The convergence of any finite neighborhood in $G$
to an appropriate tree model; $(ii)$ The conditional
independence of such a neighborhood from the residual graph, given the
neighborhood size.

The limit tree model is defined by the following sampling procedure,
yielding a $t$-generations rooted random tree $\Tree(t)$.
If $t=0$, $\Tree(t)$ is the trivial tree consisting of a single variable
node. 
For $t\ge 1$, start from a distinguished root variable node $i$ and connect it
to $l$ function nodes, whereby $l$ is a Poisson random variable with mean 
$\gamma\alpha$. For each such function nodes $a$, 
draw an independent Poisson$(\gamma)$ random variable $k_a$ and connect it to
$k_a$ new variable nodes. Finally, for each of the `first generation'
variable node $j$, sample an independent random tree distributed as 
$\Tree(t-1)$, and attach it by the root to $j$.

\begin{propo}[Convergence to random tree]\label{prop:tree_convergence}
Let $\Ball(i,t)$  be the radius-$t$ neighborhood of
any fixed variable node $i$ 
in a  random graph $G \ed\graph(n,\alpha n,\gamma/n)$, 
and $\Tree(t)$ the random tree defined above.

Given any (labeled) tree $\Tree_*$,
we write $\Ball(i,t) \simeq \Tree_*$ if $\Tree_*$ is obtained 
by the depth-first relabeling of $\Ball(i,t)$ 
following a pre-established 
convention\footnote{For instance, one might agree to preserve
the original lexicographic order among siblings.}.
Then 
\begin{eqnarray}
\lim_{n\to\infty}\prob\{\Ball(i,t)\simeq \Tree_*\}
=\prob\{\Tree_t\simeq \Tree_*\}\, .
\end{eqnarray}
\end{propo}

\begin{propo}[Bound on the neighborhood size]\label{prop:BoundNeighborhood}
Let $\Ball(i,t)$ 
be the radius-$t$ neighborhood of any fixed variable node $i$
in a random bipartite graph $G\ed \graph(n,\alpha n,\gamma/n)$,
and denote by $|\Ball(i,t)|$ its size (number of variable and function nodes). 
Then, for any $\lambda>0$ 
there exists $C(\lambda,t)$ such that, for any $n$, $M\ge 0$
\begin{eqnarray}
\prob\{|\Ball(i,t)|\ge M\}\le C(\lambda,t)\, \lambda^{-M}\, .
\end{eqnarray}
\end{propo}
\begin{proof}
Let us generalize our definition of neighborhood as follows.
If $t$ is integer, we let $\Ball(i,t+1/2)$ be the subgraph including
$\Ball(i,t)$ together with all the function nodes that have at 
least one neighbor in $\Ball(i,t)$ (as well as the edges to $\Ball(i,t)$).
We also let $\dBall(i,t+1/2)$ be the set of function nodes
that have at least one neighbor in $\Ball(i,t)$ and at least one outside.

Imagine to explore $\Ball(i,t)$ in breadth-first fashion.
For each $t$, $|\Ball(i,t+1/2)|-|\Ball(i,t)|$ is upper bounded by the sum of 
$|\dBall(i,t)|$ iid binomial random variables 
counting the number of neighbors of each node in $\dBall(i,t)$,
which are not in $\Ball(i,t)$.
For $t$ integer (respectively, half-integer), each such variables 
is stochastically dominated by a binomial with parameters $n\alpha$ 
(respectively, $n$) and $\gamma/n$. 
Therefore $|\Ball(i,t)|$ is stochastically
dominated by $\sum_{s=0}^{2t}Z_n(s)$, where $\{Z_n(t)\}$
is a Galton-Watson process with offspring distribution 
Binom$(n,\ogamma/n)$ and $\ogamma = \gamma\, \max(1,\alpha)$. 

By Markov inequality
\begin{eqnarray*}
\prob\{|\Ball(i,t)|\ge M\}\le g^n_{2t}(\lambda)
\, \lambda^{-M}\, ,\;\;\;\;\;\;\;\;
g^n_{t}(\lambda)\equiv\E\{\lambda^{\sum_{s=0}^t Z_n(s)}\}.
\end{eqnarray*}
By elementary branching processes theory 
$g^n_{t}(\lambda)$
satisfies the recursion $g^n_{t+1}(\lambda) = \lambda\xi_n(g^n_t(\lambda))$,
$g^n_0(\lambda) = \lambda$, with $\xi_n(\lambda) = \lambda(1+
2\ogamma(\lambda-1)/n)^n$. The thesis follows by 
$g^n_t(\lambda)\le  g_t(\lambda)$, where $g_t(\lambda)$ is defined as 
$g_t^n(\lambda)$ but replacing
$\xi_n(\lambda)$ with $\xi(\lambda) = e^{2\gamma(\lambda-1)}\ge 
\xi_n(\lambda)$.
\end{proof}
\begin{propo}\label{prop:poisson_independance}
Let $G=(V,F,E)$ be a random bipartite graph 
from the ensemble $\graph(n,m,p)$. Then, conditional on 
$\Ball(i,t) = (V(i,t),F(i,t),E(i,t))$,
$\cBall(i,t)$ is a random bipartite graph 
on variable nodes $V\setminus V(i,t-1)$,  function nodes
$F\setminus F(i,t)$ and same edge probability $p$.
\end{propo}
\begin{proof}
Condition on $\Ball(i,t) =  (V(i,t),F(i,t),E(i,t))$, and let 
$\Ball(r,t-1) =  (V(i,t-1),F(i,t-1),E(i,t-1))$
(notice that this is uniquely determined from $\Ball(i,t)$).
This is equivalent to conditioning on a given edge realization for
any two vertices $k$, $a$ such that $k\in V(i,t)$ and $a\in F(i,t)$.

On the other hand, $\cBall(i,t)$ is the graph with variable nodes set
 $\overline{V}\equiv V\setminus V(i,t-1)$,  function nodes
$\overline{F}\equiv F\setminus F(i,t)$, and edge set $(k,a)\in G$ such that
$k\in \overline{V}$, $a\in\overline{F}$.
Since this set of vertices couples is disjoint
from the one we are conditioning upon, and by independence of
edges in $G$,  the claim follows.
\end{proof}
%
%
%
\section{Proof of Theorem \ref{thm:ApproximateMarginal} (BP equations)}
\label{sec:ApproximateMarginalProof}

The proof of Theorem \ref{thm:ApproximateMarginal} hinges on the properties of
the random factor graph $G$ discussed in the previous Section
as well as on the correlation structure unveiled
by Theorem \ref{thm:Correlation}.
%
%
\subsection{The effect of changing $G$}

The first need to estimate the effect on changing the graph $G$ on 
marginals.
\begin{lemma}\label{lemma:PhysDegr}
Let $X$ be a random variable taking values in $\cX$ and assume
$X\to G\to Y_1\to B$ and $X\to G\to Y_1\to B$ to be Markov chains
(here $G$, $Y_{1,2}$ and $B$ are arbitrary random variables,
where $G$ stands for \emph{good} and $B$ for \emph{bad}).
Then
\begin{eqnarray}
\E\big|\big|\prob\{X\in\,\cdot\, | Y_1\}-\prob\{X\in\,\cdot\, | Y_2\}
\big|\big|_{\sTV}\le 2\, \E\big|\big|\prob\{X\in\,\cdot\, | G\}-
\prob\{X\in\,\cdot\, | B\}
\big|\big|_{\sTV}\, .
\end{eqnarray}
\end{lemma}
\begin{proof}
First consider a single Markov Chain $X\to G\to Y\to B$.
Then, by convexity of the total variation distance,
\begin{eqnarray}
\E\big|\big|\prob\{X\in\,\cdot\, | Y\}-\prob\{X\in\,\cdot\, | B\}
\big|\big|_{\sTV} & =&  \E\left|\left|
\E\left\{\prob\{X\in\,\cdot\, | G,Y\}\Big| Y\right\}-\prob\{X\in\,\cdot\, | B\}
\right|\right|_{\sTV}\le\\
&\le&   \E\left|\left|
\prob\{X\in\,\cdot\, | G,Y\}-\prob\{X\in\,\cdot\, | B\}
\right|\right|_{\sTV} = \\
&=&   \E\left|\left|
\prob\{X\in\,\cdot\, | G\}-\prob\{X\in\,\cdot\, | B\}
\right|\right|_{\sTV}\, .
\end{eqnarray}
The thesis is proved by applying this bound to both
chains $X\to G\to Y_1\to B$ and $X\to G\to Y_1\to B$,
and using triangular inequality.
\end{proof}

The next lemma estimates the effect of removing one variable node from the
graph. Notice that the graph $G$ is non-random.
\begin{lemma}\label{lemma:RemoveV}
Consider two observation systems associated to graphs
$G = (V,F,E)$ and $G' = (V',F',E')$ whereby $V=V'\setminus \{j\}$,
$F=F'$ and $E=E'\setminus \{(j,b):\, b\in\partial j\}$. 
Denote the corresponding observations as $(Y,Z(\theta))$ and
$(Y',Z'(\theta))$.
Then there exist a coupling of the observations such that, for any
$i\in V$:
\begin{align}
\E ||\prob\{X_i\in\,\cdot\, |Y,Z(\theta)\}-
\prob\{X_i\in\,\cdot\, |Y',Z'(\theta)\}
||_{\sTV}\le&\label{eq:RemoveV}\\ 
4\,\E\Big|\Big| 
\prob_{i,\partial^2j}
\{\,\cdots\, &| Y_{F\setminus \dj},Z(\theta)\}-\prod_{l\in\{i,\partial^2j\}}
\prob_{l}\{\, \cdot\,| Y_{F\setminus \dj},Z(\theta)\}
|\Big|\Big|_{\sTV}\, ,\nonumber
\end{align}
where $\partial^2j \equiv\{l\in V : d(i,l)=1\}$ and used the shorthand 
$\prob_U\{\cdots|Y_{F\setminus \dj},Z(\theta)\}$ for
$\prob\{X_U\in \cdots|Y_{F\setminus \dj},Z(\theta)\}$.

The coupling consists in sampling $X= \{X_i:\, i\in V\}$
from its (iid) distribution and then $(Y,Z(\theta))$ and $(Y',Z'(\theta))$
as observations of this configuration $X$, in such a way that
$Z(\theta)=Z'(\theta)$ and $Y_a=Y'_a$ for any $a\in F$ such that 
$\da\in V$.
\end{lemma}
\begin{proof}
Let us describe the coupling more explicitly. 
First  sample $X = \{X_i:\, i\in V\}$ and
$X'=\{X'_i:\, i\in V'\}$ in such a way that $X_i=X'_i$ for
any $i\in V$. Then, for any $i\in V$, sample $Z_i(\theta)$, $Z_i'(\theta)$
conditionally on $X_i=X'_i$ in such a way that $Z_i(\theta)=Z_i'(\theta)$.
Sample $Z'_j(\theta)$ conditionally on $X_j'$.
For any $a\in F$ such that $\da\in V$,
sample $Y_a$, $Y_a'$
conditionally on $X_{\da}=X'_{\da}$ in such a way that 
$Y_a=Y'_a$. Finally for $a\in\dj$, sample 
$Y_a$, $Y_a'$ independently, 
conditional on $X_{\da}\neq X'_{\da}$

Notice that the following are Markov Chains
\begin{eqnarray}
X_i\to (X_{\partial^2 j},Y,Z(\theta))\;\;\to (Y,Z(\theta))\to
(Y_{F\setminus\dj},Z(\theta))\, ,\\
X_i\to  (X_{\partial^2 j},Y,Z(\theta))\to (Y',Z'(\theta))\to
(Y_{F\setminus\dj},Z(\theta))\, .
\end{eqnarray} 
The only non-trivial step is  
$(X_{\partial^2 j},Y,Z(\theta))\to (Y',Z'(\theta))$.
Notice that, once $X_{\partial^2 j}$ is known, $Y_{\dj}$ is conditionally
independent from the other random variables.
Therefore we can produce $(Y',Z'(\theta))$ first scratching $Y_{\dj}$,
then sampling $X_j'$ independently, next sampling $Y_{\dj}'$ and 
$Z'_j(\theta)$ and finally scratching both $X'_j$ and $X_{\partial^2 j}$.

Applying Lemma \ref{lemma:PhysDegr} to the chains above, we get
\begin{align*}
\E \big|\big|\prob\{X_i\in\,\cdot\, |Y,Z(\theta)\}-
\prob\{X_i\in\,\cdot\, |Y',Z'(\theta)\}
\big|\big|_{\sTV}\le& \\
\le 2\, \E 
\big|\big|\prob\{X_i\in\,\cdot\,|Y_{F\setminus\dj},Z(\theta)\}&
-\prob\{X_i\in\,\cdot\,|X_{\partial^2 j},Y,Z(\theta)\}\big|\big|_{\sTV} =\\
= 2\, \E 
\big|\big|\prob\{X_i\in\,\cdot\,|Y_{F\setminus\dj},Z(\theta)\}&
-\prob\{X_i\in\,\cdot\,|X_{\partial^2 j},Y_{F\setminus\dj},Z(\theta)\}\big|\big|_{\sTV}\, ,
\end{align*}
where in the last step, we used the fact that $Y_{\dj}$ is conditionally 
independent of $X_i$, given  $X_{\partial^2 j}$.
The thesis is proved using the identity (valid for any two random variables
$U,W$)
\begin{eqnarray}
\E||\prob\{U\in\,\cdot\,\}-\prob\{U\in\,\cdot\,|W\}||_{\sTV}  = 
||\prob\{(U,W)\in\,\cdots\,\}-\prob\{U\in\,\cdot\,\}\prob\{W\in\,\cdot\,
\}||_{\sTV}\, ,
\end{eqnarray}
and the bound (that follows from triangular inequality)
\begin{align*}
||\prob\{(U,W_1\dots W_k)\in\,\cdots\,\}-\prob\{U\in\,\cdots\,\}\prob\{(W_1\dots W_k)\in\,\cdots\,\}||_{\sTV}\le&\\
\le 2\,
||\prob\{(U,W_1\dots W_k)\in\,\cdots\,\}-\prob\{U\in\,\cdots\,\}
\prob\{W_1\in\,\cdot\,\}\cdots&\prob\{W_k\in\,\cdot\,\} ||_{\sTV}\, .
\end{align*}
\end{proof}

An analogous Lemma estimates the effect of removing a function node.
\begin{lemma}\label{lemma:RemoveF}
Consider two observation systems associated to graphs
$G = (V,F,E)$ and $G' = (V',F',E')$ whereby $V=V'$,
$F=F'\setminus \{a\}$ and $E=E'\setminus \{(j,a):\, j\in\partial a\}$. 
Denote the corresponding observations as $(Y,Z(\theta))$ and
$(Y',Z'(\theta))$, with $Z(\theta)=Z'(\theta)$ and $Y=Y'\setminus\{a\}$.
Then, for any $i\in V$:
\begin{align}
\E ||\prob\{X_i\in\,\cdot\, |Y,Z(\theta)\}-
\prob\{X_i\in\,\cdot\, |Y',Z'(\theta)\}
||_{\sTV}\le&\label{eq:RemoveF}\\ 
4\,\E\Big|\Big| 
\prob_{i,\partial a}
\{\,\cdots\, &| Y_{F\setminus \da},Z(\theta)\}-\prod_{l\in\{i,\da\}}
\prob_{l}\{\, \cdot\,| Y_{F\setminus \da},Z(\theta)\}
|\Big|\Big|_{\sTV}\, .\nonumber
\end{align}
where we used the shorthand 
$\prob_U\{\cdots|Y_{F\setminus a},Z(\theta)\}$ for
$\prob\{X_U\in \cdots|Y_{F\setminus a},Z(\theta)\}$.
\end{lemma}
\begin{proof}
The proof is completely analogous (and indeed easier)
to the one of Lemma \ref{lemma:RemoveV}.
It is sufficient to consider the Markov chain
$X_i\to (X_{\da},Y,Z(\theta))\to (Y,Z(\theta))\to (Y_{F\setminus a},
Z(\theta))$, and bound the total variation distance considered here
in terms of the first and last term in the chain, where we notice that
$(Y_{F\setminus a},Z(\theta)) = (Y',Z'(\theta))$.
We omit details to avoid redundancies.
\end{proof}

Next, we study the effect of removing a variable node from a random 
bipartite graph.
\begin{lemma}\label{lemma:AvRemoveV}
Let $G=(V,F,E)$ and $G'=(V',F',E')$ be two random graphs
from, respectively, the $\graph(n-1,\alpha n,\gamma/n)$
and $\graph(n,\alpha n,\gamma/n)$ ensembles. 
Consider two information systems on such graphs. Let
$(Y,Z(\theta))$ and $(Y',Z'(\theta))$ be the corresponding 
observations, and $\mu^{\theta}_i$, $\mu^{\theta}_i{}'$
the conditional distributions of $X_i$ in the two systems.

It is then possible to couple $G$ to $G'$ and, for each $\theta$
$(Y,Z(\theta))$ to $(Y',Z'(\theta))$ and choose 
a constant ${\sf C} ={\sf C}(|\cX|,\alpha,\gamma)$
(bounded uniformly for $\gamma$ and $1/\alpha$ bounded), such that,
for any $\eps>0$ and any $i\in V\cap V'$,
\begin{eqnarray}
\int_{0}^\eps\E_{G}\E ||\mu^{\theta}_i-\mu^{\theta}_i{}'||_{\sTV} \le 
\frac{{\sf C}}{\sqrt{n}}\, .
\end{eqnarray}
Further, such a coupling can be produced by letting 
$V' = V\cup \{n\}$, $F'=F$ and $E'=E\cup \{(n,a):a\in\partial n\}$
where $a\in\dn$ independently with probability $\gamma/n$. Finally
$(Y,Z(\theta))$ and $(Y',Z'(\theta))$ are coupled as in Lemma
\ref{lemma:RemoveV}. 
\end{lemma}
\begin{proof}
Take $V=[n-1]$, $V'=[n]$, $F=F'=[n\alpha]$ and sample the edges 
by letting, for any $i\in [n-1]$, $(i,a)\in E$ if and only if $(i,a)\in E'$.
Therefore $E = E'\setminus\{(n,a):\, a\in \dn\}$ (here $\dn$ is 
the neighborhood of variable node $n$ with respect to 
the edge set $E'$). Coupling $(Y,Z(\theta))$
and  $(Y',Z'(\theta))$ as in Lemma \ref{lemma:RemoveV}, 
and using the bound proved there, we get
\begin{align}
\int_{0}^\eps\E_{G}\E ||\mu^{\theta}_i-\mu^{\theta}_i{}'||_{\sTV}
\le 4\int_{0}^{\eps} \E_G\E \Big|\Big| 
\prob_{i,\d2n}
\{\,\cdots\, &| Y_{F\setminus \dn},Z(\theta)\}-\prod_{l\in\{i,\d2n\}}
\prob_{l}\{\, \cdot\,| Y_{F\setminus \dn},Z(\theta)\}
|\Big|\Big|_{\sTV}\de\theta\, ,\label{eq:FirstAvRemoveV}
\, .
\end{align}

In order to estimate the total variation distance on the right hand side,
we shall condition on $|\dn|$ and $|\d2n|$. Once this is done,
the conditional probability
$\prob_{i,\d2n}
\{\,\cdots\, | Y_{F\setminus \dn},Z(\theta)\}$
is distributed as the conditional probability of 
$|\d2n|+$ variables, in a system $\hG(|\dn|)$ with $n-1$ variable nodes
and $n\alpha-|\dn|$ function nodes. Let us  denote
by $(\hY,\hZ(\theta))$ the corresponding observations
(and by $\hprob$, $\hE$  probability and expectations).  
Then the right hand side in Eq.~(\ref{eq:FirstAvRemoveV}) is equal to
\begin{align*}
 4\int_{0}^{\eps} \E_{|\dn|,|\d2n|}&\E_G\Big\{
\E \Big|\Big| 
\prob_{i,\d2n}
\{\,\cdots\, | Y_{F\setminus \dn},Z(\theta)\}-\prod_{l\in\{i,\d2n\}}
\prob_{l}\{\, \cdot\,| Y_{F\setminus \dn},Z(\theta)\}
|\Big|\Big|_{\sTV}
\Big|\,|\dn|,|\d2n|\Big\}\, \de\theta =\\ 
&= 4\int_{0}^{\eps}  \E_{|\dn|,|\d2n|}\E_{\hG(|\dn|)}
\hE \Big|\Big| 
\hprob_{1\dots|\d2n|+1}
\{\,\cdots\, | \hY,\hZ(\theta)\}-\prod_{l=}^{|\d2n|}
\hprob_{l}\{\, \cdot\,| \hY,\hZ(\theta)\}
|\Big|\Big|_{\sTV}\, \de\theta\le\\
&\le 4\E_{|\dn|,|\d2n|}\left\{(|\cX|+1)^{|\d2n+1|}\sqrt{H(X_1)\eps/(n-1)}
\right\}
+4\eps\prob\{|\d2n|+1\ge \sqrt{n}/10\}\ ,.
\end{align*}
In the last step we applied Corollary \ref{coro:Correlation}
and distinguished the cases $|\d2n|+1\ge \sqrt{n}/10$
(then bounding the total variation distance by $1$) and $|\d2n|+1< \sqrt{n}/10$
(then bounding $A_{n-1,|\d2n|+1}$ by $\sqrt{2}$ thanks to the
estimate in Theorem \ref{thm:Correlation}).
The thesis follows using Proposition \ref{prop:BoundNeighborhood}
to bound both terms above (notice in fact that $|\d2n|\le |\Ball(i,1)|$).
\end{proof}

Again, an analogous estimate holds for the effect of removing
one function node. The proof is omitted as it is almost identical to the 
previous one.
\begin{lemma}\label{lemma:AvRemoveF}
Let $G=(V,F,E)$ and $G'=(V',F',E')$ be two random graphs
from, respectively, the $\graph(n,\alpha n-1,\gamma/n)$
and $\graph(n,\alpha n,\gamma/n)$ ensembles. 
Consider two information systems on such graphs. Let
$(Y,Z(\theta))$ and $(Y',Z'(\theta))$ be the corresponding 
observations, and $\mu^{\theta}_i$, $\mu^{\theta}_i{}'$
the conditional distributions of $X_i$ in the two systems.

It is then possible to couple $G$ to $G'$ and, for each $\theta$
$(Y,Z(\theta))$ to $(Y',Z'(\theta))$ and choose 
a constant ${\sf C} ={\sf C}(|\cX|,\alpha,\gamma)$
(bounded uniformly for $\gamma$ and $1/\alpha$ bounded), such that,
for any $\eps>0$ and any $i\in V\cap V'$,
\begin{eqnarray}
\int_{0}^\eps\E_{G}\E ||\mu^{\theta}_i-\mu^{\theta}_i{}'||_{\sTV} \le 
\frac{{\sf C}}{\sqrt{n}}\, .
\end{eqnarray}
Further, such a coupling can be produced by letting 
$V'= V$, $F'=F\setminus\{a\}$, for a fixed function node $a$, 
and $E'=E\cup \{(j,a):j\in\partial a\}$
where $j\in\da$ independently with probability $\gamma/n$. Finally
$(Y,Z(\theta))$ and $(Y',Z'(\theta))$ are coupled as in Lemma
\ref{lemma:RemoveF}. 
\end{lemma}
%
%
\subsection{BP equations}
\label{sec:ApproximateMarginalActualProof}

We begin by proving a useful technical Lemma.
\begin{lemma}\label{lemma:Simple}
Let $p_1$, $p_2$ be probability distribution over a finite set $\cS$,
and $q:\hcS\times\cS\to\reals_+$ be a non-negative function.
Define, for $a=1,2$ the probability distributions
\begin{eqnarray}
\hp_a(x) \equiv\frac{\sum_{y\in\cS}q(x,y)\,p_a(y)}
{\sum_{x'\in\hcS,y'\in\cS}q(x',y')\,p_a(y')}\, .
\end{eqnarray}
Then
\begin{eqnarray}
||\hp_1-\hp_2||_{\sTV} \le 2\,\left(\frac{\max_{y\in\cS}\sum_x
q(x,y)}{\min _{y\in\cS}\sum_x q(x,y)}\right)\; ||p_1-p_2||_{\sTV}\, .
\end{eqnarray}
\end{lemma}
\begin{proof}
Using the inequality $|(a_1/b_1)-(a_2/b_2)|\le |a_1-a_2|/b_1+(a_2/b_2)|b_1-
b_2|/b_1$
(valid for $a_1,a_2,b_1,b_2\ge 0$), we get
\begin{eqnarray*}
|\hp_1(x)-\hp_2(x)|\le\frac{\sum_y q(x,y)|p_1(y)-p_2(y)|}
{\sum_{x',y'}q(x',y')p_1(y')}+\frac{\sum_{y}q(x,y)p_2(y)}{\sum_{x',y'}
q(x',y')p_2(y')}\, \frac{\left|\sum_{x',y'}q(x',y')(p_1(y')-p_2(y'))\right|}
{\sum_{x',y'}q(x',y')p_1(y')}\, .
\end{eqnarray*}
Summing over $x$ we get
\begin{eqnarray*}
||\hp_1-\hp_2||_{\sTV}\le \frac{\sum_y (\sum_x q(x,y))|p_1(y)-p_2(y)|}
{\sum_{y'}(\sum_{x'}q(x',y'))p_1(y')}\, ,
\end{eqnarray*}
whence the thesis follows.
\end{proof}

Given a graph $G$,   $i\in V$, $t\ge 1$, we let $\Ball\equiv\Ball(i,t)$,
$\cBall\equiv \cBall(i,t)$ and $\dBall\equiv \dBall(i,t)$,
Further, we introduce the shorthands
\begin{eqnarray}
W_{\Ball} & \equiv 
&\{ Y_a :  \, \da\subseteq\Ball,\,\da\not\subseteq\dBall\}\cup\{Z_i :\, i\in\Ball\setminus \dBall\}\, ,\label{eq:WNotation1}\\
W_{\cBall} & \equiv 
&\{ Y_a :  \, \da\subseteq\cBall\}\cup\{Z_i :\, i\in\cBall\}\, .\label{eq:WNotation2}
\end{eqnarray}
Notice that $W_{\Ball}$, $W_{\cBall}$ form a partition of the variables
in $Y,Z(\theta)$. Further $W_{\Ball}$, $W_{\cBall}$ are conditionally 
independent given $X_{\dBall}$.
As a consequence, we have the following simple bound.
\begin{lemma}\label{lemma:Elementary}
For any two non-negative functions $f$ and $g$, we have
\begin{eqnarray}
\E\{f(W_{\Ball})\, g(W_{\cBall})\}\le \max_{x_{\Ball}}
\E\{f(W_{\Ball})|X_{\dBall}=x_{\dBall}\}\, \E\{g(W_{\cBall})\}\, .
\end{eqnarray}
\end{lemma}
\begin{proof}
Using the conditional independence property we have
\begin{eqnarray*}
\E\Big\{f(W_{\Ball})\, g(W_{\cBall})\Big\} = 
\E\{\E[f(W_{\Ball})|X_{\dBall}]\, \E[g(W_{\cBall})|X_{\dBall}]\Big\}
\le \max_{x_{\Ball}}
\E\{f(W_{\Ball})|X_{\dBall}=x_{\dBall}\}\, \E\Big\{
\E[g(W_{\cBall})|X_{\dBall}]\Big\}\, ,
\end{eqnarray*}
which proves our claim.
\end{proof}

It is easy to see that the conditional distribution
of $(X_i,X_{\dBall})$ takes the form (with an abuse of notation
we write $\prob\{X_U|\, \cdots\,\}$ instead of $\prob\{X_U=x_U|\, \cdots\,\}$)
\begin{eqnarray}
\prob\{X_i,X_{\dBall}|Y,Z(\theta)\} & = &\frac{\prob\{X_i,W_{\Ball}|X_{\dBall},
W_{\cBall}\}\,\prob\{X_{\dBall}|W_{\cBall}\}}{\sum_{X'_i,X'_{\dBall}}
\prob\{X'_i,W_{\Ball}|X'_{\dBall},W_{\cBall}\}\,\prob\{X'_{\dBall}|W_{\cBall}\}}
=\\
& = & \frac{\prob\{X_i,W_{\Ball}|X_{\dBall}\}\,
\prob\{X_{\dBall}|W_{\cBall}\}}{\sum_{X'_i,X'_{\dBall}}
\prob\{X'_i,W_{\Ball}|X'_{\dBall}\}\,\prob\{X'_{\dBall}|W_{\cBall}\}}\, .
\end{eqnarray}
If $\Ball$ is a small neighborhood of $i$, the most intricate 
component in the above formulae is the probability 
$\prob\{X_{\dBall}|W_{\cBall}\}$. It would be nice if we could replace
this term by the product of the marginal probabilities of $X_j$, for 
$j\in\dBall$. We thus define 
\begin{eqnarray}
\qprob\{X_i,X_{\dBall}||Y,Z(\theta)\} & = &
\frac{\prob\{X_i,W_{\Ball}|X_{\dBall}\}\,
\prod_{j\in\dBall}\prob\{X_j|W_{\cBall}\}}{\sum_{X'_i,X'_{\dBall}}
\prob\{X'_i,W_{\Ball}|X'_{\dBall}\}\,
\prod_{j\in\dBall}\prob\{X'_j|W_{\cBall}\}}\, .
\end{eqnarray}
Notice that this is a probability kernel, but not a conditional probability
(to stress this point we used the double separator $||$). 

Finally, we recall the definition of local marginal 
$\mu_i^{\theta}(x_i)$ and 
introduce, by analogy, the approximation $\mu_i^{\theta,t}(\,\cdot\,)$
\begin{eqnarray}
\mu_i^{\theta}(x_i)= \sum_{X_{\dBall}} 
\prob\{X_i=x_i,X_{\dBall}|Y,Z(\theta)\}\, ,\;\;\;\;\;\;\;
\mu_i^{\theta,t}(x_i)\equiv \sum_{X_{\dBall}} 
\qprob\{X_i=x_i,X_{\dBall}|Y,Z(\theta)\}\, .\label{def:Mut}
\end{eqnarray}
It is easy to see that, for $t=1$, $\mu_i^{\theta,t}$ is nothing 
but the result of applying belief propagation to the
marginals of the neighbors of $i$ with respect to the reduced graph that 
does not include $i$. Formally, in the notation of Theorem
\ref{thm:ApproximateMarginal}:
\begin{eqnarray}
\mu_i^{\theta,1}(x_i) = \Fbp_i(\{\mu_{j\to a}^{\theta}\}_{a\in\di, j\in\da
\setminus i})(x_i)\, .
\end{eqnarray}

The result below shows that indeed the boundary condition 
on $X_{\dBall}$ can be chosen as factorized, thus providing a more general 
version of Theorem \ref{thm:ApproximateMarginal}.
\begin{thm}[BP equations, more general version]
\label{lemma:ApproximateMarginal}
Consider an observations system on a random bipartite graph $G = (V,F,E)$ from
the $\graph(n,\alpha n,\gamma/n)$ ensemble, and assume the
noisy observations to be $M$-soft. Then there exists a constant 
${\sf A}$ depending on $t,\alpha,\gamma,M,|\cX|,\eps$, such that
for any $i\in V$, and any $n$
\begin{eqnarray}
\int_0^{\eps}\E_G\E||\mu_i^{\theta}-\mu_i^{\theta,t}||_{\sTV}
\;\de\theta \le \frac{{\sf A}}{\sqrt{n}}\, .
\end{eqnarray}
\end{thm}
\begin{proof}
Notice that the definitions of $\mu_i^{\theta}$, $\mu_i^{\theta,t}$ 
have the same form as $\hp_1$, $\hp_2$ in Lemma \ref{lemma:Simple},
whereby $x$ corresponds to $X_i$ and $y$ to $X_{\dBall}$.
We have therefore
\begin{eqnarray}
||\mu_i^{\theta}-\mu_i^{\theta,t}||_{\sTV} \le
2\left(\frac{\max_{x_{\dBall}}\prob\{W_{\Ball}|X_{\dBall}=x_{\dBall}\}}
{\min_{x_{\dBall}}\prob\{W_{\Ball}|X_{\dBall}=x_{\dBall}\} }\right)\,
\Big|\Big| \prob\{X_{\dBall}=\,\cdot\,|W_{\cBall}\}-
\prod_{j\in\dBall}\prob\{X_j=\,\cdot\,|W_{\cBall}\}\Big|\Big|_{\sTV}\, .
\end{eqnarray}
Given observations $Z(\theta)$ and $U\subseteq V$, let us denote
as $\Comp(Z(\theta),U)$ the values of $x_U$ such that,
for any $i\in U$ with $Z_{i}(\theta) = (Z_i,x^0_i)$ with $x^0_i\neq \ast$,
one has $x_i = x^0_i$ (i.e. the set of assignments $x_U$ that are compatible
with direct observations).
Notice that the factor in parentheses can be upper bounded
as
\begin{eqnarray}
\frac{\max_{x_{\dBall}}\sum_{x_{\Ball\setminus\dBall}}
\prob\{W_{\Ball}|X_{\Ball}=x_{\Ball}\}}
{\min_{x_{\dBall}}\sum_{x_{\Ball\setminus\dBall}}
\prob\{W_{\Ball}|X_{\Ball}=x_{\Ball}\} }
&\le&
\frac{\max_{x_{\dBall}}\max_{x_{\Ball\setminus\dBall}
\in\Comp(Z(\theta),\Ball\setminus\dBall)}
\prob\{W_{\Ball}|X_{\Ball}=x_{\Ball}\}}
{\min_{x_{\dBall}}\min_{x_{\Ball\setminus\dBall}\in
\Comp(Z(\theta),\Ball\setminus\dBall)}
\prob\{W_{\Ball}|X_{\Ball}=x_{\Ball}\} }
=\\
&\le&\frac{\max_{x_{\Ball}
\in\Comp(Z(\theta),\Ball\setminus\dBall)}
\prob\{W_{\Ball}|X_{\Ball}=x_{\Ball}\}}
{\min_{x_{\Ball}\in
\Comp(Z(\theta),\Ball\setminus\dBall)}
\prob\{W_{\Ball}|X_{\Ball}=x_{\Ball}\} }\, .
\end{eqnarray}

Using Lemma \ref{lemma:Elementary} to take expectation with respect to 
the observations $(Y,Z(\theta))$, we get
\begin{eqnarray}
\E||\mu_i^{\theta}-\mu_i^{\theta,t}||_{\sTV} \le C(\Ball)
\,\E\left|\left| \prob\{X_{\dBall}=\,\cdot\,|W_{\cBall}\}-
\prod_{j\in\dBall}\prob\{X_j=\,\cdot\,|W_{\cBall}\}\right|\right|_{\sTV}\, ,
\label{eq:ExpNoise}
\end{eqnarray}
where (with the shorthand $\prob\{\,\cdot\,|x_U\}$ for 
$\prob\{\,\cdot\,|X_U=x_U\}$, and omitting the arguments from $\Comp$,
since they are clear from the context)
\begin{eqnarray*}
C(\Ball) &=& 2\max_{x'_{\dBall}}\E\left\{\left.
\left(\frac{\max_{x_{\dBall}}\prob\{W_{\Ball}|x_{\dBall}\}}
{\min_{x_{\dBall}}\prob\{W_{\Ball}|x_{\dBall}\} }\right)
\right|X_{\dBall}=x'_{\dBall}\right\}\le\\
&\le& 2\max_{x'_{\Ball}}\E\left\{\left.
\left(\frac{\max_{x_{\Ball}
\in\Comp}
\prob\{W_{\Ball}|X_{\Ball}=x_{\Ball}\}}
{\min_{x_{\Ball}\in
\Comp}
\prob\{W_{\Ball}|X_{\Ball}=x_{\Ball}\} }
\right)
\right|X_{\Ball}=x'_{\Ball}\right\}\le\\
&\le& 2\max_{x'_{\Ball}}\E\left\{\left.
\prod_{a\in\Ball}\frac{\max_{x_{\da}}\prob\{Y_a|x_{\da}\}}
{\min_{x_{\da}}\prob\{Y_a|x_{\da}\}}
\prod_{i\in \Ball\setminus \dBall}
\frac{\max_{x_i\in\Comp}\prob\{Z_i(\theta)|x_{i}\}}
{\min_{x_{i}\in\Comp}\prob\{Z_i(\theta)|x_{i}\}}
\right|X_{\Ball}=x'_{\Ball}\right\}\le\\
&\le & 
\prod_{a\in\Ball}\max_{x'_{\da}}
\E\left\{\left.\frac{\max_{x_{\da}}\prob\{Y_a|x_{\da}\}}
{\min_{x_{\da}}\prob\{Y_a|x_{\da}\}}\right|X_{\da}=x'_{\da}\right\}
\prod_{i\in \Ball\setminus \dBall}
\max_{x'_{i}}\E\left\{\left.
\frac{\max_{x_i}\prob\{Z_i|x_{i}\}}
{\min_{x_{i}}\prob\{Z_i|x_{i}\}}
\right|X_{i}=x'_{i}\right\}
\le  M^{|\Ball|}\, .
\end{eqnarray*}
In the last step we used the hypothesis of soft noise, and before
we changed $Z_i(\theta)$ in $Z_i$ because the difference is
irrelevant under the restriction $x\in\Comp$, and subsequently 
removed this restriction.

We now the expectation of Eq.~(\ref{eq:ExpNoise}) over the random graph 
$G$, conditional on $\Ball$
\begin{eqnarray}
\E_G\left\{\left.
\E||\mu_i^{\theta}-\mu_i^{\theta,t}||_{\sTV}\right|\Ball\right\}
\le M^{|\Ball|}
\E_G\left\{\left.
\E\left|\left| \prob\{X_{\dBall}=\,\cdot\,|W_{\cBall}\}-
\prod_{j\in\dBall}\prob\{X_j=\,\cdot\,|W_{\cBall}\}\right|\right|_{\sTV}
\right|\Ball\right\}
\, ,
\end{eqnarray}
Notice that the the conditional expectation is equivalent 
to an expectation over a random graph on variable nodes
$(V\setminus V(\Ball))\cup \dBall$, and function nodes $F\setminus F(\Ball)$
(where $V(\Ball)$ and $F(\Ball)$ denotes the variable and function node sets
of $\Ball$).
The distribution of this `residual graph' is the same as for the original 
ensemble: for any $j\in (V\setminus V(\Ball))\cup \dBall$
and any $b\in F\setminus F(\Ball)$, the edge $(j,b)$ is included 
independently with probability $\gamma/n$.
We can therefore apply Corollary \ref{coro:Correlation}
\begin{eqnarray}
\int_0^{\eps}\E_G\left\{\left.
\E||\mu_i^{\theta}-\mu_i^{\theta,t}||_{\sTV}\right|\Ball\right\}
\;\de\theta\le M^{|\Ball|}(|\cX|+1)^{|\Ball|}
A_{n-|\Ball|,|\Ball|}\sqrt{H(X_1)\eps/(n-|\Ball|)}\, .
\end{eqnarray}

We can now take expectation over $\Ball =\Ball(i,t)$, and invert expectation 
and integral over $\theta$, since the integrand is non-negative and bounded.
We single out the case $|\Ball|>\sqrt{n}/10$ and upper bound the
total variation distance by $1$ in this case. 
In the case $|\Ball|\le\sqrt{n}/10$
we upper bound $A_{|\Ball|,n-|\Ball|}$ by $\sqrt{2}$ and lower bound
$n-|\Ball|$ by $n/2$, thus yielding, for $\widetilde{M}\equiv M(1+\cX)$:
\begin{eqnarray}
\int_0^{\eps}\E_G
\E||\mu_i^{\theta}-\mu_i^{\theta,t}||_{\sTV}
\;\de\theta &\le &\sqrt{4H(X_1)\eps/n}\;\E\{\widetilde{M}^{|\Ball|}\}+
\prob\{|\Ball|>\sqrt{n}/10\}\, .
\end{eqnarray}
The thesis follows by applying Proposition \ref{prop:BoundNeighborhood}
to both terms.
\end{proof}
%
%
%
\section{Proof of Theorem \ref{thm:Main} (density evolution)}
\label{sec:MainProof}

Given a probability distribution $S$
over $\meas^2(\cX)$, we define the probability distribution
$\sP_{S,k}$ over $\meas(\cX)\times \cdots\times \meas(\cX)$
($k$ times) by
\begin{eqnarray}
\sP_{S,k}\left\{(\mu_1,\dots,\mu_k)\in A\right\}
=\int \rho^{k}(A)\,  S(\de\rho)\, .
\end{eqnarray}
where $\rho^{k}$ is the law of $k$ iid random $\mu$'s with common distribution
$\rho$.
We shall denote by $\sE_{S,k}$ expectation with respect to the same measure.
\begin{lemma}\label{lemma:deFinetti}
Let $G=(V,F,E)$ be any bipartite graph whose distribution is invariant under 
permutation of the variable nodes in $V=[n]$, and
$\mu_i(\,\cdot\,) \equiv \prob\{ X_i=\cdot|Y,Z\}$ the marginals of an
observations system on $G$. Then, for any diverging sequence 
$R_0\subseteq\naturals$ there exists a subsequence $R$ and  
a distribution $S$ on
$\meas^2(\cX)$ such that, for any subset $\{i(1),\dots,i(k)\}\subseteq [n]$
of distinct variable nodes, and any bounded Lipschitz function 
$\varphi:\meas(\cX)^k\to\reals$:
\begin{eqnarray}
\lim_{n\in R}\E_G\E\{\varphi(\mu_{i(1)},\dots,\mu_{i(k)})\} =
\sE_{S,k}\{\varphi(\mu_{1},\dots,\mu_{k})\}\, .\label{eq:deFinetti}
\end{eqnarray}
\end{lemma}
\begin{proof}
We shall assume, without loss of generality, that $R_0=\naturals$.
Notice that $(\mu_1,\dots,\mu_n)$ is a family of exchangeable random 
variables. By tightness, for each $i=1,2,\dots$, there exists
a subsequence $R_i$ such that $(\mu_1,\dots,\mu_i)$ converges in distribution,
and $R_{i+1}\subseteq R_i$.
Construct the subsequence $R$ whose $j$-th element is the $j$-th
element of $R_j$.
Then for any $k$, 
$(\mu_{i(1)},\dots,\mu_{i(k)})$ converges in distribution along $R$
to an exchangeable set  $(\mu^{(k)}_{1},\dots,\mu^{(k)}_{k})$.
Further the projection of the law
of $(\mu^{(k)}_1,\dots,\mu^{(k)}_k)$ on the first $k-1$ variables is the law
of  $(\mu^{(k-1)}_1,\dots,\mu^{(k-1)}_{k-1})$.
Therefore, this defines an exchangeable distribution over
the infinite collection of random variables $\{\mu_i:\, i=1,2,\dots\}$.
By de Finetti, Hewitt-Savage Theorem  \cite{deFinetti,HewittSavage}
there exists $S$ such that, for any $k$, the joint distribution 
of $(\mu_1,\dots,\mu_k)$ is $\sP_{S,k}$. In particular
\begin{eqnarray*}
\lim_{n\in R}\E_G\E\{\varphi(\mu_{i(1)},\dots,\mu_{i(k)})\} =
\E\{\varphi(\mu_{1},\dots,\mu_{k})\}=
\sE_{S,k}\{\varphi(\mu_{1},\dots,\mu_{k})\}\, .
\end{eqnarray*}
\end{proof}

\begin{proof}\emph{[Main Theorem]}
By Lemma \ref{lemma:deFinetti}, Eq.~(\ref{eq:Main}) holds for some
probability distribution $S_{\theta}$ on $\meas^2(\cX)$. It remains 
to prove that $S_{\theta}$ is supported over the fixed points of the
density evolution equation (\ref{eq:DensityEvolution}).

Let $\varphi:\meas(\cX)\to\reals$ be a test function that we can assume,
without loss of generality bounded by $1$, and with Lipschitz constant $1$. 
Further, let $\dBall(i)\equiv\dBall(i,1)$ and $\mu^{\theta(i)}_{\dBall(i)}
\equiv \{\mu_j^{\theta(i)};\; j\in\dBall(i)\}$.
By Theorem \ref{thm:ApproximateMarginal}, together with
the Lipschitz  property and boundedness, we have
\begin{eqnarray}
\int_0^{\eps}\E_G\E\,\left\{\Big[
\varphi\big(\mu_i^{\theta}\big)-
\varphi\big(\Fbp_i\big(\mu^{\theta,(i)}_{\dBall(i)}\big)\big)\Big]^2
\right\}
\;\de\theta \le \frac{{\sf A}'}{\sqrt{n}}\, .
\end{eqnarray}
Fix now two variable nodes, say $i=1$ and $i=2$.
Using Cauchy-Schwarz, this implies
\begin{eqnarray*}
\int_0^{\eps}\left|\E_G\E\,\left\{\Big[
\varphi\big(\mu_1^{\theta}\big)-
\varphi\big(\Fbp_1\big(\mu^{\theta,(1)}_{\dBall(1)}\big)\big)\Big]\,
\Big[
\varphi\big(\mu_2^{\theta}\big)-
\varphi\big(\Fbp_2\big(\mu^{\theta,(2)}_{\dBall(2)}\big)\big)\Big]
\right\}\right|
\;\de\theta \le  \frac{{\sf A}'}{\sqrt{n}}\, .
\end{eqnarray*}
Applying dominated convergence theorem, it follows that, 
for almost all $\theta\in[0,\eps]$,
\begin{eqnarray}
\lim_{n\to\infty} \E_G\E\,\left\{\Big[
\varphi\big(\mu_1^{\theta}\big)-
\varphi\big(\Fbp_1\big(\mu^{\theta,(1)}_{\dBall(1)}\big)\big)\Big]\,
\Big[
\varphi\big(\mu_2^{\theta}\big)-
\varphi\big(\Fbp_2\big(\mu^{\theta,(2)}_{\dBall(2)}\big)\big)\Big]
\right\} = 0\, .\label{eq:LimitScalarProd}
\end{eqnarray}
By Lemma \ref{lemma:deFinetti}, we can find a sequence 
$R_{\theta}$, and a distribution $S_{\theta}$ over $\meas(\cX)^2$,
such that Eq.~(\ref{eq:deFinetti}) holds. We claim that along such a sequence
\begin{eqnarray}
&&\hspace{-0.75cm}\lim_{n\in R} \E_G\E\,\{
\varphi\big(\mu_1^{\theta}\big)\varphi\big(\mu_1^{\theta}\big)\}=
\sE_{S_\theta,2}\{\varphi(\mu_1)\varphi(\mu_2)\}\, ,\label{eq:Limit1}\\
&&\hspace{-0.75cm}
\lim_{n\in R} \E_G\E\,\big\{\varphi\big(\mu_1^{\theta}\big)
\varphi\big(\Fbp_2\big(\mu^{\theta,(2)}_{\dBall(2)}\big)\big\}=
\E\,\sE_{S_\theta,k+1}\{\varphi\big(\mu_1\big)
\varphi\big(\Fde\big(\mu_2,\cdots,\mu_{k+1}\big))\}\, ,\label{eq:Limit2}\\
&&\hspace{-0.75cm}\lim_{n\in R} \E_G\E\,\left\{
\varphi\big(\Fbp_1\big(\mu^{\theta,(1)}_{\dBall(1)}\big)\big)
\varphi\big(\Fbp_2\big(\mu^{\theta,(2)}_{\dBall(2)}\big)\big)\right\}=
\E\,
\sE_{S_\theta,k_1+k_2}\{\varphi\big(\Fde_1\big(\mu_1,\cdots,\mu_{k_1}\big)\big)
\varphi\big(\Fde_2\big(\mu_{k_1+1},\cdots,\mu_{k_1+k_2}\big))\}\, .\nonumber\\
\label{eq:Limit3}
\end{eqnarray}
Here the expectations on the right hand sides are 
with respect to marginals $\mu_1,\mu_2,\dots$ distributed according to 
$\sP_{S_{\theta},\cdot}$ (this expectation is denoted as 
$\sE_{S_{\theta},\cdot}$)
as well as with respect to independent random mappings 
$\Fde:\meas(\cX)^*\to\meas(\cX)$ defined 
as in Section \ref{sec:BP}, cf. Eq.~(\ref{eq:DensityEvolution})
(this includes expectation with respect to $k$, $k_1$, $k_2$
and is denoted as $\E$).

Before proving the above limits, let us show that they implies the thesis.
Substituting Eqs.~(\ref{eq:Limit1}) to (\ref{eq:Limit3}) in 
Eq.~(\ref{eq:LimitScalarProd}) and re-ordering the terms we get
\begin{eqnarray}
&&\int\Delta(\rho)^2\; S_{\theta}(\de\rho) = 0\, , \\
&&\Delta(\rho)\equiv \int \varphi(\mu)\,\rho(\de\mu)
-\E\int \varphi(\Fde\big(\mu_1,\cdots,\mu_{k}\big))\; \rho(\de\mu_1)
\cdots \rho(\de\mu_k)\, .
\end{eqnarray}
Therefore $\Delta(\rho)=0$ $S_{\theta}$-almost surely, which is 
what we needed to show in order to prove Theorem \ref{thm:Main}.

Let us now prove the limits above. Equation~(\ref{eq:Limit1})
follows is an immediate consequence of Lemma \ref{lemma:deFinetti}.
Next consider Eq.~(\ref{eq:Limit2}), 
and condition the expectation on the left-hand side upon 
$\Ball(i=2,t=1) = \Ball$,  as well as upon $W_{\Ball}$, cf. 
Eq.~(\ref{eq:WNotation1}). First notice that, by Lemma \ref{lemma:AvRemoveV}
\begin{eqnarray}
\E_G\E\,\big\{||\mu_1^{\theta}-\mu_1^{\theta,(2)}||_{\sTV}|\Ball,W_\Ball\big\}
\le \frac{{\sf C}}{\sqrt{n-|\Ball|}}\, .
\end{eqnarray}
and condition the expectation on the left-hand side upon 
$\Ball(i,1) = \Ball$, as well as upon $W_{\Ball}$, cf. 
Eq.~(\ref{eq:WNotation1}). 
As a consequence, by Lipschitz property and boundedness of $\varphi$
\begin{eqnarray}
\left|\E_G\E\,\big\{\varphi\big(\mu_1^{\theta}\big)
\varphi\big(\Fbp_2\big(\mu^{\theta,(2)}_{\dBall(2)}\big)|\Ball,W_\Ball\big\}
-\E_G\E\,\big\{\varphi\big(\mu_1^{\theta,(2)}\big)
\varphi\big(\Fbp_2\big(\mu^{\theta,(2)}_{\dBall(2)}\big)|\Ball,W_\Ball\big\}
\right| \le \frac{{\sf C}}{\sqrt{n-|\Ball|}}\, .\label{eq:Subst}
\end{eqnarray}
In the second term the $\mu^{\theta,(2)}_j$ are independent of 
the conditioning, of the function  $\Fbp_2$ (which is deterministic once
$\Ball$, $W_{\Ball}$ are given). Therefore, by Lemma \ref{lemma:deFinetti}
(here we are taking the limit on the joint distribution
of the $\mu_j^{\theta(2)}$, but not on $\Fbp_2$; to emphasize this point 
we note the latter as $\Fbps_2$)
\begin{eqnarray}
\lim_{n\in R_{\theta}} \E_G\E\,\big\{\varphi\big(\mu_1^{\theta,(2)}\big)
\varphi\big(\Fbps_2\big(\mu^{\theta,(2)}_{\dBall(2)}\big)|\Ball,W_\Ball\big\}
= \sE_{S,k}\{\varphi(\mu_1)\varphi(\Fbps_2(\mu_{2},\dots,
\mu_{1+|\dBall(2)}|))\}\, .
\end{eqnarray}
(Notice that the graph whose expectation is considered on the left hand
side is from the ensemble $\graph(n-|V(\Ball)|,\alpha n-|F(\Ball)|,
\gamma/n)$. The limit measure $S_{\theta}$ could \emph{a priori} be 
different from the one for the  ensemble $\graph(n,\alpha n,
\gamma/n)$. However, Lemmas \ref{lemma:AvRemoveV}, \ref{lemma:AvRemoveF}
 imply that this cannot be the case.)

By using dominated convergence and Eq.~(\ref{eq:Subst}) we get 
\begin{eqnarray*}
\lim_{n\in R_{\theta}}
\E_G\E\,\big\{\varphi\big(\mu_1^{\theta}\big)
\varphi\big(\Fbps_2\big(\mu^{\theta,(2)}_{\dBall(2)}\big)\} &=&
\E_{\Ball,W_{\Ball}}\left\{\lim_{n\in R_{\theta}}
\E_G\E\,\big\{\varphi\big(\mu_1^{\theta}\big)
\varphi\big(\Fbps_2\big(\mu^{\theta,(2)}_{\dBall(2)}\big)|\Ball,W_\Ball\big\}
\right\}
=\\
&= &\E_{\Ball,W_{\Ball}}\sE_{S,k}
\Big\{\varphi(\mu_1)\varphi(\Fbps_2(\mu_{2},\dots,
\mu_{1+|\dBall(2)}|))\Big\}\, .
\end{eqnarray*}
Finally we can take the limit $n_*\to \infty$ as well. By local 
convergence of the graph to the tree model, we have uniform convergence of
$\Fbps_2$ to $\Fde$ and thus Eq.~(\ref{eq:Limit2}).

The proof of Eq.~(\ref{eq:Limit3}) is completely analogous to the latter
and is omitted to avoid redundancies.
\end{proof}
%
%

\bibliographystyle{alpha}

\end{document}